\documentclass{article}
\usepackage[a4paper, total={6in, 8in}]{geometry}
\usepackage[utf8]{inputenc}
\usepackage[english]{babel}
\usepackage{amsmath,amssymb}
\usepackage{amsthm}
\usepackage{color}
\usepackage{xcolor}
\usepackage{enumerate}
\usepackage[numbers,sort&compress]{natbib}
\usepackage{tikz,pgfplots}
\usepackage[ruled,vlined]{algorithm2e}
\usetikzlibrary{intersections, backgrounds}
\usepgfplotslibrary{groupplots}

\usepackage{todonotes}

\newtheorem{theorem}{Theorem}[section]
\newtheorem{lemma}{Lemma}[section]
\newtheorem{definition}{Definition}[section]
\newtheorem{proposition}{Proposition}[section]
\newtheorem{remark}{Remark}[section]

\newcommand{\T}{\mathsf{T}}
\renewcommand{\H}{\mathsf{H}}
\newcommand{\cN}{\mathcal{N}}
\newcommand{\dd}{\mathrm{d}}
\newcommand{\range}{\mathrm{im}}
\newcommand{\R}{\mathbb{R}}
\newcommand{\N}{\mathbb{N}}
\newcommand{\C}{\mathbb{C}}
\newcommand{\ri}{\mathrm{i}}
\newcommand{\hinf}{{\mathcal{H}_\infty}}
\newcommand{\rhinf}{\mathcal{RH}_\infty}
\newcommand{\htwo}{{\mathcal{H}_2}}

\DeclareMathOperator{\tr}{tr}

\definecolor{colorone}{rgb}{0.0,0.0,1.0}
\definecolor{colortwo}{rgb}{0.4889,0.2356,0.1781}
\definecolor{colorthree}{rgb}{0.3644,0.2441,0.4243}
\definecolor{ourcolor}{rgb}{0.2422,0.9433,0.3044}

\newcommand{\nx}{n_x}
\newcommand{\enu}{n_u}
\newcommand{\ny}{n_y}
\newcommand{\pHsys}{{\Sigma_{\mathrm{pH}}}}

\newcommand{\mechSys}{{\Sigma_{\mathrm{sso}}}}
\newcommand{\pHtf}{{G_{\mathrm{pH}}}}
\newcommand{\mechtf}{{G_{\mathrm{sso}}}}
\newcommand{\pHsysr}{{\Sigma_{\mathrm{pH},r}}}
\newcommand{\mechSysr}{{\Sigma_{\mathrm{sso},r}}}

\newcommand{\pJ}{J}
\newcommand{\pR}{R}
\newcommand{\pQ}{Q}
\newcommand{\pB}{B}
\newcommand{\pM}{M}
\newcommand{\pD}{D}
\newcommand{\pK}{K}
\newcommand{\cP}{\mathcal{P}}
\newcommand{\cQ}{\mathcal{Q}}
\newcommand{\cR}{\mathcal{R}}

\newcommand{\cL}{\mathcal{L}}
\newcommand{\loss}{L}
\newcommand{\LSQ}{\mathrm{lsq}}

\renewcommand{\epsilon}{\varepsilon}

\DeclareMathOperator{\vtf}{vtf}
\DeclareMathOperator{\ftv}{ftv}
\DeclareMathOperator{\vtu}{vtu}
\DeclareMathOperator{\utv}{utv}
\DeclareMathOperator{\vtsu}{vtsu}
\DeclareMathOperator{\sutv}{sutv}
\DeclareMathOperator{\Real}{Re}
\renewcommand{\hat}{\widehat}

\newcommand{\dta}{\mathrm{d}\theta_J}
\newcommand{\dtb}{\mathrm{d}\theta_R}
\newcommand{\dtc}{\mathrm{d}\theta_Q}
\newcommand{\dtd}{\mathrm{d}\theta_B}

\newcommand{\dtM}{\mathrm{d}\theta_M}
\newcommand{\dtD}{\mathrm{d}\theta_D}
\newcommand{\dtK}{\mathrm{d}\theta_K}
\newcommand{\dtB}{\mathrm{d}\theta_B}

\newcommand{\inp}{\pB}
\newcommand{\dyn}{\mathcal{F}_0}

\newcommand{\dynmech}{\mathcal{F}_0}
\newcommand{\vv}{\widehat v}
\newcommand{\uu}{\widehat u}

\newcommand{\dJ}{Y_1}
\newcommand{\dQ}{Y_2}

\newcommand{\cI}{\mathcal{I}}

\newcommand{\errorTFwidth}{0.5\textwidth}
\newcommand{\errorTFheight}{0.25\textheight}

\usepackage[a4paper,breaklinks,colorlinks]{hyperref}
\hypersetup{
     linkcolor = red,
     citecolor = green,
     urlcolor = blue,
}

\title{SOBMOR: Structured Optimization-Based Model Order Reduction\thanks{This work is supported by the German Research Foundation (DFG) within the project VO2243/2-1: ``Interpolationsbasierte numerische Algorithmen in der robusten Regelung" and the DFG Cluster of Excellence MATH+ within the project AA4-5: ``Energy-based modeling, simulation, and optimization of power systems under uncertainty". This research has mainly been carried out while the second author was affiliated with Universit\"at Hamburg and Technische Universit\"at Berlin. Their support is gratefully acknowledged.}}
\author{Paul Schwerdtner\thanks{\emph{Corresponding author}, Technische Universit\"at Berlin, Institut f\"ur Mathematik, Stra{\ss}e des 17. Juni 136, 10623 Berlin, Germany. E-Mail: schwerdt@math.tu-berlin.de.} and Matthias Voigt\thanks{UniDistance Suisse, \"Uberlandstrasse 12, 3900 Brig, Switzerland. E-Mail: matthias.voigt@fernuni.ch.}}

\begin{document}

\maketitle
\begin{abstract}
  Model order reduction (MOR) methods that are designed to preserve structural features of a given full order model (FOM) often suffer from a lower accuracy when compared to their non-structure-preserving counterparts. In this paper, we present a framework for structure-preserving MOR, which allows to compute structured reduced order models (ROMs) with a much higher accuracy. The framework is based on parameter optimization, i.\,e., the elements of the system matrices of the ROM are iteratively varied to minimize an objective functional that measures the difference between the FOM and the ROM. The structural constraints can be encoded in the parametrization of the ROM. The method only depends on frequency response data and can thus be applied to a wide range of dynamical systems.

  We illustrate the effectiveness of our method on a port-Hamiltonian and on a symmetric second-order system in a comparison with other structure-preserving MOR algorithms.


\end{abstract}

\section{Introduction}

The increased demand for analysis and control of interconnected systems has promoted the investigation of structured models that encode physical properties of the underlying system such as stability or passivity. This is useful as these properties are hard to check or to enforce for system interconnections but are often automatically given in structured models. When (sub)-systems describe complex physical phenomena, high-fidelity modeling often leads to models with a large state-space dimension which renders simulations or controller synthesis computationally challenging or even impossible.

This problem is addressed by model order reduction (MOR). MOR allows to compute low order surrogate models that approximate the original model's input/output map, i.\,e., the same excitation of the full order model (FOM) and the reduced order model (ROM) leads to a similar response and the ROM can be used in simulations or for controller design instead of the FOM. In the setting of interconnected systems, it is essential to not only mimic the given input/output map of some component but also preserve its physical structure to be able to still benefit from the automatic properties of structured models.

For that, structure-preserving MOR methods have been introduced. Throughout this paper we consider the MOR of port-Hamiltonian (pH) systems (see \cite{VanJ14} for an overview) of the form
\begin{align*}
  \pHsys:\;
  \begin{cases}
    \dot x(t)=(J-R)Q x(t) + B u(t),\\
    y(t) = B^\T Q x(t),
  \end{cases}
\end{align*}
where $J,\, R,\, Q \in \R^{\nx \times \nx}$ and $B \in \R^{\nx \times \enu}$ 
with the structural constraints $J = -J^\T$, $R \succeq 0$, and $Q \succeq 0$. Here, for a matrix $X \in \R^{n\times n}$, $X \succ 0$ ($X \succeq 0$) indicates that $X$ is symmetric positive (semi)-definite. Note that often $Q \succ 0$ is assumed, e.\,g., in  \cite{BeaMX15} but for simplicity of presentation we assume that $Q \succeq 0$. The \emph{Hamiltonian} $H:\R^n \to \overline{\R^+} :=\{ a\in\R\; | \;a\ge0 \}$ of $\pHsys$ is given by
\begin{align*}
  H(x(t))=\frac{1}{2}x(t)^{\T}Qx(t).
\end{align*}
It describes the total internal energy of the system in the state $x(t)$.


The other system class we consider consists of symmetric second-order systems (SSO) with co-located force actuators and position sensors which usually arise in the modeling of mechanical systems. They are of the form
\begin{align*}
  \mechSys:\;
  \begin{cases}
    M \ddot x(t) + D\dot x(t)+ K x(t) = B u(t),\\
    \phantom{M \ddot x(t) + D\dot x(t)+ K}y(t) = B^{\T} x(t),
  \end{cases}
\end{align*}
where $M,\, D,\, K \in \R^{\nx \times \nx}$ and $B\in \R^{\nx \times \enu}$ and where $M\succeq 0$, $D \succeq 0$, and $K\succeq 0$. Such systems are discussed, e.\,g., in \cite{Reis2008,PetL10}. In a similar fashion, we could consider systems with co-located velocity instead of position outputs. The latter can be shown to be pH by an appropriate transformation to first-order form.


Our goal is to compute ROMs which share exactly the same structural features as the FOMs. For a pH system this means that we want to obtain a ROM of the form
\begin{align*}
  \pHsysr:\;
  \begin{cases}
    \dot x_r(t)=(J_r-R_r)Q_r x_r(t) + B_r u(t),\\
         y_r(t) = B_r^\T Q_r x_r(t),
  \end{cases}
\end{align*}
where $J_r,\, R_r,\, Q_r \in \R^{r \times r}$, $B_r \in \R^{r \times \enu}$ with  $J_r = -J_r^\T$, $R_r \succeq 0$, and $Q_r \succeq 0$ and with $r \ll \nx$. Analogously, a ROM for the system $\mechSys$ would be of the form
\begin{align} \label{eq:mechrom}
  \mechSysr:\;
  \begin{cases}
    M_r \ddot x_r(t) + D_r\dot x_r(t)+ K_r x_r(t) = B_r u(t),\\
    \phantom{M_r \ddot x_r(t) + D_r\dot x_r(t)+ K_r}y_r(t) = B_r^{\T} x_r(t),
  \end{cases}
\end{align}
where $M_r,\, D_r,\, K_r \in \R^{r \times r}$ and $B_r\in \R^{r \times \enu}$ with $M_r\succeq0$, $D_r \succeq 0$, $K_r \succeq 0$, and $r \ll \nx$.

In the literature there exist several different approaches to MOR and multiple different algorithms have been developed. Popular classes of reduction algorithms for linear time-invariant systems are the following: (i) eigenvalue-based approaches such as \emph{modal truncation} in which dominant eigenvalues of an underlying eigenvalue problem are extracted; (ii) balancing-based methods such as \emph{balanced truncation (BT)} in which input and output energies are considered to identify states which contribute only little to the input/output behavior of the system; and (iii) rational approximation methods such as \emph{moment matching} or the \emph{iterative rational Krylov algorithm (IRKA)} in which approximation techniques for the system's transfer function are employed. We refer to the textbooks \cite{Ant05} and \cite{AntBG20} for an overview.

Structure-preserving variants of the above mentioned methods have been 
considered intensively in the literature. For example, the works 
\cite{Polyuga2008,Polyuga2010,WolLEK10,Gugercin2012,HauMM19} treat the case of 
pH systems, while 
\cite{MeyS96,Sal05,BaiS05,ChaLVV06,RomM06,SalL06,Reis2008,HarVS10,BenKS13,
PanWL13,SaaSW19,DorRV20} consider the case of second-order systems. Furthermore, recent advances in the backward error analysis of the structured eigenvalue problems associated with the system classes under consideration \cite{Dop18,Meh20} are of potential use in the derivation of eigenvalue-based reduction methods for pH and SSO systems.  
Note that pH systems are also \emph{passive}. In rough terms, this means that such systems cannot internally produce energy. The preservation of this property is also addressed in many articles such as \cite{Obe91,GuiO13,Sor05,IonRA08}. 

All the methods mentioned above have in common that they are projection-based, i.\,e., the involved matrices are reduced by an appropriate 
multiplication with (tall and skinny) projection matrices. 
In contrast to that we develop a new framework which is based on making an 
ansatz for the structure of the ROM and then optimizing the entries of the 
resulting ROM parametrization. We call this approach {SOBMOR}, which stands for \textbf{S}tructured \textbf{O}ptimization-\textbf{B}ased \textbf{M}oder \textbf{O}rder \textbf{R}eduction.
For the optimization we introduce a least 
squares type objective functional which measures the differences between the 
transfer functions of the FOM and the ROM.

\begin{figure}[t]
  \centering
  \begin{tabular}{c}
    \begin{tikzpicture}[/tikz/background rectangle/.style={fill={rgb,1:red,1.0;green,1.0;blue,1.0}, draw opacity={1.0}}, show background rectangle]
\begin{axis}[title={}, title style={at={{(0.5,1)}}, font={{\fontsize{14 pt}{18.2 pt}\selectfont}}, color={rgb,1:red,0.0;green,0.0;blue,0.0}, draw opacity={1.0}, rotate={0.0}}, legend style={color={rgb,1:red,0.0;green,0.0;blue,0.0}, draw opacity={1.0}, line width={1}, solid, fill={rgb,1:red,1.0;green,1.0;blue,1.0}, fill opacity={1.0}, text opacity={1.0}, font={{\fontsize{8 pt}{10.4 pt}\selectfont}}, at={(1.02, 1)}, anchor={north west}}, axis background/.style={fill={rgb,1:red,1.0;green,1.0;blue,1.0}, opacity={1.0}}, anchor={north west}, xshift={1.0mm}, yshift={-1.0mm}, width={0.8\textwidth}, height={0.25\textheight}, scaled x ticks={false}, xlabel={$r$}, x tick style={color={rgb,1:red,0.0;green,0.0;blue,0.0}, opacity={1.0}}, x tick label style={color={rgb,1:red,0.0;green,0.0;blue,0.0}, opacity={1.0}, rotate={0}}, xlabel style={, font={{\fontsize{11 pt}{14.3 pt}\selectfont}}, color={rgb,1:red,0.0;green,0.0;blue,0.0}, draw opacity={1.0}, rotate={0.0}}, xmajorgrids={true}, xmin={3.52}, xmax={20.48}, xtick={{4.0,8.0,12.0,16.0,20.0}}, xticklabels={{$4$,$8$,$12$,$16$,$20$}}, xtick align={inside}, xticklabel style={font={{\fontsize{8 pt}{10.4 pt}\selectfont}}, color={rgb,1:red,0.0;green,0.0;blue,0.0}, draw opacity={1.0}, rotate={0.0}}, x grid style={color={rgb,1:red,0.0;green,0.0;blue,0.0}, draw opacity={0.1}, line width={0.5}, solid}, axis x line*={left}, x axis line style={color={rgb,1:red,0.0;green,0.0;blue,0.0}, draw opacity={1.0}, line width={1}, solid}, scaled y ticks={false}, ylabel={$\hinf$ error}, y tick style={color={rgb,1:red,0.0;green,0.0;blue,0.0}, opacity={1.0}}, y tick label style={color={rgb,1:red,0.0;green,0.0;blue,0.0}, opacity={1.0}, rotate={0}}, ylabel style={, font={{\fontsize{11 pt}{14.3 pt}\selectfont}}, color={rgb,1:red,0.0;green,0.0;blue,0.0}, draw opacity={1.0}, rotate={0.0}}, ymode={log}, log basis y={10}, ymajorgrids={true}, ymin={4.169570177021927e-6}, ymax={0.7106243024441214}, ytick={{1.0e-5,0.0001,0.001,0.01,0.1}}, yticklabels={{$10^{-5}$,$10^{-4}$,$10^{-3}$,$10^{-2}$,$10^{-1}$}}, ytick align={inside}, yticklabel style={font={{\fontsize{8 pt}{10.4 pt}\selectfont}}, color={rgb,1:red,0.0;green,0.0;blue,0.0}, draw opacity={1.0}, rotate={0.0}}, y grid style={color={rgb,1:red,0.0;green,0.0;blue,0.0}, draw opacity={0.1}, line width={0.5}, solid}, axis y line*={left}, y axis line style={color={rgb,1:red,0.0;green,0.0;blue,0.0}, draw opacity={1.0}, line width={1}, solid}, colorbar style={title={}}, point meta max={nan}, point meta min={nan}]
    \addplot[color=colorone, name path={1c97c6ec-a4e8-4e04-b96d-a14d666e4a38}, draw opacity={1.0}, line width={1}, solid, mark={*}, mark size={3.0 pt}, mark options={color={rgb,1:red,0.0;green,0.0;blue,0.0}, draw opacity={1.0}, fill=colorone, fill opacity={1.0}, line width={0.75}, rotate={0}, solid}]
        coordinates {
            (4,0.15753775527276784)
            (6,0.04027322931391762)
            (8,0.009039266088716613)
            (10,0.0013965612409578953)
            (12,0.00036038074920490167)
            (14,9.760860161700115e-5)
            (16,3.718258464869135e-5)
            (18,1.4904130437195738e-5)
            (20,9.515525824944514e-6)
        }
        ;
    \addlegendentry {BT}
    \addplot[color=blue, name path={fdc86282-b401-4979-abae-4468f1aaf760}, draw opacity={1.0}, line width={1}, solid, mark={*}, mark size={3.0 pt}, mark options={color={rgb,1:red,0.0;green,0.0;blue,0.0}, draw opacity={1.0}, fill=blue, fill opacity={1.0}, line width={0.75}, rotate={0}, solid}, dashed]
        coordinates {
            (4,0.4813471942945101)
            (6,0.5053326346940811)
            (8,0.33265402950983736)
            (10,0.2792712148058059)
            (12,0.2535149168194241)
            (14,0.1870396561695439)
            (16,0.14977502520076516)
            (18,0.12700565585039805)
            (20,0.086946101277349)
        }
        ;
    \addlegendentry {pH-BT}
    \addplot[color=purple, name path={8bcb0643-48fc-42f0-bc69-1dd8a3c976d3}, draw opacity={1.0}, line width={1}, solid, mark={diamond*}, mark size={3.0 pt}, mark options={color={rgb,1:red,0.0;green,0.0;blue,0.0}, draw opacity={1.0}, fill=purple, fill opacity={1.0}, line width={0.75}, rotate={0}, solid}]
        coordinates {
            (4,0.17027228081798673)
            (6,0.12484574367356895)
            (8,0.04303408150106201)
            (10,0.0018477730009792128)
            (12,0.0025173078149703)
            (14,0.00011531401245745243)
            (16,0.0003000306047994916)
            (18,2.2769436542425816e-5)
            (20,1.0259198015791216e-5)
        }
        ;
    \addlegendentry {IRKA}
    \addplot[color=purple, name path={5e7e1fa2-9a18-4027-b9a8-5d99613a1c14}, draw opacity={1.0}, line width={1}, solid, mark={diamond*}, mark size={3.0 pt}, mark options={color=black, draw opacity={1.0}, fill=purple, fill opacity={1.0}, line width={0.75}, rotate={0}, solid}, dashed]
        coordinates {
            (4,0.29644543479395774)
            (6,0.2520497943209842)
            (8,0.16276581594136239)
            (10,0.1287240783673755)
            (12,0.09709961195728065)
            (14,0.07263289970224365)
            (16,0.05774756684325709)
            (18,0.04021116428444131)
            (20,0.027135580825011017)
        }
        ;
    \addlegendentry {pH-IRKA}
    \addplot[color=ourcolor, name path={ee35c385-852e-4ea4-bc93-4a19681d6410}, draw opacity={1.0}, line width={1}, solid, mark={triangle*}, mark size={4.0 pt}, mark options={color={rgb,1:red,0.0;green,0.0;blue,0.0}, draw opacity={1.0}, fill=ourcolor, fill opacity={1.0}, line width={0.75}, rotate={0}, solid}]
        coordinates {
          (4, 7.716396e-02)
          (6, 3.342240e-02)
          (08,6.411867e-03)
          (10,7.429191e-04)
          (12,1.685170e-04)
          (14,5.347107e-05)
          (16,1.659474e-05)
          (18,7.503271e-06)
          (20,4.931897e-06)
        }
        ;
    \addlegendentry {SOBMOR}
\end{axis}
\end{tikzpicture} \\ (a) pH MOR results \\ \begin{tikzpicture}[/tikz/background rectangle/.style={fill={rgb,1:red,1.0;green,1.0;blue,1.0}, draw opacity={1.0}}, show background rectangle]
\begin{axis}[title={}, title style={at={{(0.5,1)}}, font={{\fontsize{14 pt}{18.2 pt}\selectfont}}, color={rgb,1:red,0.0;green,0.0;blue,0.0}, draw opacity={1.0}, rotate={0.0}}, legend style={color={rgb,1:red,0.0;green,0.0;blue,0.0}, draw opacity={1.0}, line width={1}, solid, fill={rgb,1:red,1.0;green,1.0;blue,1.0}, fill opacity={1.0}, text opacity={1.0}, font={{\fontsize{8 pt}{10.4 pt}\selectfont}}, at={(1.02, 1)}, anchor={north west}}, axis background/.style={fill={rgb,1:red,1.0;green,1.0;blue,1.0}, opacity={1.0}}, anchor={north west}, xshift={1.0mm}, yshift={-1.0mm}, width={0.8\textwidth}, height={0.25\textheight}, scaled x ticks={false}, xlabel={$r$}, x tick style={color={rgb,1:red,0.0;green,0.0;blue,0.0}, opacity={1.0}}, x tick label style={color={rgb,1:red,0.0;green,0.0;blue,0.0}, opacity={1.0}, rotate={0}}, xlabel style={, font={{\fontsize{11 pt}{14.3 pt}\selectfont}}, color={rgb,1:red,0.0;green,0.0;blue,0.0}, draw opacity={1.0}, rotate={0.0}}, xmajorgrids={true}, xmin={4.52}, xmax={21.48}, xtick={{5.0,9.0,13.0,17.0,21.0}}, xticklabels={{$5$,$9$,$13$,$17$,$21$}}, xtick align={inside}, xticklabel style={font={{\fontsize{8 pt}{10.4 pt}\selectfont}}, color={rgb,1:red,0.0;green,0.0;blue,0.0}, draw opacity={1.0}, rotate={0.0}}, x grid style={color={rgb,1:red,0.0;green,0.0;blue,0.0}, draw opacity={0.1}, line width={0.5}, solid}, axis x line*={left}, x axis line style={color={rgb,1:red,0.0;green,0.0;blue,0.0}, draw opacity={1.0}, line width={1}, solid}, scaled y ticks={false}, ylabel={$\hinf$ error}, y tick style={color={rgb,1:red,0.0;green,0.0;blue,0.0}, opacity={1.0}}, y tick label style={color={rgb,1:red,0.0;green,0.0;blue,0.0}, opacity={1.0}, rotate={0}}, ylabel style={, font={{\fontsize{11 pt}{14.3 pt}\selectfont}}, color={rgb,1:red,0.0;green,0.0;blue,0.0}, draw opacity={1.0}, rotate={0.0}}, ymode={log}, log basis y={10}, ymajorgrids={true}, ymin={1e-9}, ymax={0.11435994055361993}, ytick={{1.0e-8,1.0e-6,0.0001,0.01}}, yticklabels={{$10^{-8}$,$10^{-6}$,$10^{-4}$,$10^{-2}$}}, ytick align={inside}, yticklabel style={font={{\fontsize{8 pt}{10.4 pt}\selectfont}}, color={rgb,1:red,0.0;green,0.0;blue,0.0}, draw opacity={1.0}, rotate={0.0}}, y grid style={color={rgb,1:red,0.0;green,0.0;blue,0.0}, draw opacity={0.1}, line width={0.5}, solid}, axis y line*={left}, y axis line style={color={rgb,1:red,0.0;green,0.0;blue,0.0}, draw opacity={1.0}, line width={1}, solid}, colorbar style={title={}}, point meta max={nan}, point meta min={nan}]
  \addplot[color=blue, name path={6e2877b2-05c8-46fb-955e-6f02e81a55d9}, draw opacity={1.0}, line width={1}, solid, mark={*}, mark size={3.0 pt}, mark options={color={rgb,1:red,0.0;green,0.0;blue,0.0}, draw opacity={1.0}, fill=colorone, fill opacity={1.0}, line width={0.75}, rotate={0}, solid}]
  coordinates {
    (5,  0.04124344785120114)
    (7,  0.027751421590272808)
    (9,  0.0031892155249361454)
    (11, 0.00034308205507223104)
    (13, 3.389368322841432e-5)
    (15, 3.5179019130119426e-6)
    (17, 1.5581756042523324e-7)
    (19, 3.526359560664446e-8)
    (21, 4.0871085381989925e-9)
  }
  ;
    \addlegendentry {BT}
    \addplot[color=blue, name path={6bbdc228-316c-411a-b429-b7b55f11bff0}, draw opacity={1.0}, line width={1}, solid, mark={*}, mark size={3.0 pt}, mark options={color={rgb,1:red,0.0;green,0.0;blue,0.0}, draw opacity={1.0}, fill=blue, fill opacity={1.0}, line width={0.75}, rotate={0}, solid}, dashed]
        coordinates {
            (5,0.04282119599987606)
            (7,0.019696244977096777)
            (9,0.005847487980626263)
            (11,0.0014658958365942642)
            (13,0.00033165135178661294)
            (15,7.05002174038418e-5)
            (17,1.4258904952455394e-5)
            (19,2.992117374064235e-6)
            (21,6.098886235369962e-7)
        }
        ;
    \addlegendentry {SO-BT}
    \addplot[color=ourcolor, name path={33da23d1-c61f-40e9-bfb2-5a7a3ef72ca4}, draw opacity={1.0}, line width={1}, solid, mark={triangle*}, mark size={4.0 pt}, mark options={color={rgb,1:red,0.0;green,0.0;blue,0.0}, draw opacity={1.0}, fill=ourcolor, fill opacity={1.0}, line width={0.75}, rotate={0}, solid}]
        coordinates {
            (05,4.174850e-03)
            (07,3.298406e-04)
            (09,2.184062e-05)
            (11,2.779742e-07)
            (13,2.757395e-08)
            (15,1.230419e-08)
            (17,1.113622e-08)
            (19,1.115440e-08)
            (21,1.234850e-08)
        }
        ;
    \addlegendentry {SOBMOR}
\end{axis}

\end{tikzpicture} \\ (b) SSO MOR results
  \end{tabular}
  \caption{Comparison of the $\hinf$ error of structure-preserving MOR methods with non-structure-preserving MOR methods for varying reduced state dimensions $r$. Note that we have used a first-order representation of the second-order model to compute the BT ROMs in (b). Therefore, the BT ROMs are also in first-order form unlike the SO-BT and SOBMOR ROMs.}
  \label{fig:introplot}
\end{figure}

The motivation for developing our new approach is the observation that for the same reduced order, the structure-preserving variants of the balancing and interpolatory MOR methods often lead to a lower accuracy of the input/output map than their non-structure-preserving counterparts. This is illustrated in Figure~\ref{fig:introplot}, in which we display the $\hinf$ error for different ROMs with varying reduced state-space dimensions including our new approach. The underlying models and MOR algorithms are explained in the following section. Here we want to emphasize the much lower accuracy (of several orders of magnitude) of the structure-preserving methods compared to the non-structure-preserving ones. On the other hand, our direct optimization approach allows the computation of ROMs with a prescribed model structure that have an $\hinf$ (and also $\mathcal{H}_2$) error that is closer to the respective error of ROMs created with non-structure-preserving MOR.

Our paper is organized as follows. In the next section, we discuss a few basics of MOR and briefly explain the structure-preserving methods that we use for a comparison with our approach. After that, in Section~\ref{sec:our_method}, we derive our method. In particular, we discuss parametrizations of the system structures under consideration and the gradients of the error transfer functions between the FOM and the ROM with respect to the ROM parameters in Theorems~\ref{thm::pHgradients} and~\ref{thm:mechGradients}. This enables us to use gradient-based optimization to reduce this error and to achieve a close matching between the FOM and the ROM. Certain implementation variants are introduced in Section~\ref{sec:implVar} while in Section~\ref{sec:NumExp} we test our algorithm on a couple of benchmark systems. There we will also see the superior approximation quality compared to the methods existing in the literature. 

\section{Background}
In this section we give a few details on the most popular methods for MOR of linear time-invariant systems. We also briefly outline how these methods can be adapted to preserve model structures. In particular, we give some details on the methods we use for benchmarking our approach.
\label{sec:background}

\subsection{Model Order Reduction Overview}
In the general linear MOR framework we consider a FOM of the form
\begin{align*}
  \Sigma:\;
  \begin{cases}
    \dot x(t) = Ax(t) + B u(t),\\
         y(t) = Cx(t) + D u(t),
  \end{cases}
\end{align*}
where $A\in \R^{\nx \times \nx}$, $B\in\R^{\nx \times \enu}$, $C\in\R^{\ny \times \nx}$, and $D \in \R^{\ny \times \enu}$. The goal of MOR is finding a ROM 
\begin{align*}
  \Sigma_r:\;
  \begin{cases}
    \dot x_r(t) = A_rx_r(t) + B_r u(t),\\
         y_r(t) = C_rx_r(t) + D_r u(t),
  \end{cases}
\end{align*}
where $A_r\in \R^{r \times r}$, $B_r\in\R^{r \times \enu}$, $C_r\in\R^{\ny \times r}$, and $D_r \in \R^{\ny \times \enu}$ and $r \ll \nx$. Moreover, $\| y-y_r \|$ should be small relative to all admissible inputs $u(\cdot)$ in an appropriate function space and if $\Sigma$ is asymptotically stable (equivalently, all eigenvalues of $A$ have negative real part), one typically requests that also $\Sigma_r$ is asymptotically stable \cite{Ant05}.

To measure the reduction error one usually considers the transfer functions  \begin{equation*}
     G(s) := C(sI_{n_x} - A)^{-1} B + D, \quad G_r(s) = C_r(sI_r - A_r)^{-1} B_r + D_r.
 \end{equation*}
 If $\Sigma$ and $\Sigma_r$ are asymptotically stable, then $G$ and $G_r$ are elements of the \emph{Hardy spaces}
 \begin{align*}
  \mathcal{H}_{2}^{\ny \times \enu} &:= \left\{ H : \C^+ \to \C^{\ny \times \enu} \; \bigg| \; H \text{ is analytic and } \sup_{\sigma > 0} \int_{-\infty}^\infty {\|H(\sigma+\ri\omega)\|}_{\rm F}^2 \mathrm{d}\omega < \infty \right\}, \\
  \mathcal{H}_{\infty}^{\ny \times \enu} &:= \left\{ H : \C^+ \to \C^{\ny \times \enu} \; \bigg| \; H \text{ is analytic and } \sup_{\lambda \in \C^+} {\|H(\lambda)\|}_2 < \infty \right\},
 \end{align*}
 where $\C^+ := \{ \lambda \in \C \;|\; \Real(\lambda) > 0 \}$.
 Since in our setting, both $G$ and $G_r$ are real-rational functions, we also make use of the real-rational subspaces of $\mathcal{H}_{2}^{\ny \times \enu}$ and $\mathcal{H}_{\infty}^{\ny \times \enu}$ which we denote by $\mathcal{RH}_{2}^{\ny \times \enu}$ and $\mathcal{RH}_{\infty}^{\ny \times \enu}$, respectively.
 Since any function $G \in \mathcal{RH}_{2}^{\ny \times \enu}$ or $G \in \mathcal{RH}_{\infty}^{\ny \times \enu}$ is analytic in $\C^+$ and rational, it can be uniquely analytically extended to the imaginary axis and thus, $G(\ri \omega)$ is well-defined for all $\omega \in \R$. The spaces $\mathcal{RH}_{2}^{\ny \times \enu}$ and $\mathcal{RH}_{\infty}^{\ny \times \enu}$ are normed spaces ($\mathcal{RH}_{2}^{\ny \times \enu}$ is even an inner product space) and are equipped with the norms
 \begin{equation}
  \left\| G \right\|_{\mathcal{H}_2} := \left(\frac{1}{2\pi}\int_{-\infty}^\infty {\|G(\ri\omega)\|}_{\rm F}^2 \mathrm{d}\omega\right)^{1/2}, \quad \left\| G \right\|_{\mathcal{H}_\infty} := \sup_{\omega \in \R} {\|G(\ri\omega)\|}_2.
 \end{equation}
 With these concepts at hand, the error of the reduction procedure can by quantified by $\left\| G-G_r \right\|_{\mathcal{H}_2}$ or $\left\| G -G_r\right\|_{\mathcal{H}_\infty}$ and can often either be optimized or bounded by the reduction algorithms.

We briefly mention the ideas of some of the most popular reduction algorithms, namely \emph{balanced truncation (BT)} and the \emph{iterative rational Krylov algorithm (IRKA)}, see the textbooks \cite{Ant05,AntBG20}. Note that both methods assume that the FOM is asymptotically stable.
BT is based on energy functionals related to the system. For some state $x_0\in\R^{\nx}$ these functionals measure the input energy needed to steer the state from zero to $x_0$ (when this is possible) and the energy that can be extracted from the output, if the initial state is $x_0$. Both energy functionals are quadratic forms of $x_0$ which are described by the symmetric and positive semidefinite controllability and observability Gramians $\cP,\,\cQ \in \R^{\nx \times \nx}$. In BT one computes a balancing transformation by an invertible matrix $T\in\R^{\nx \times \nx}$ such that the transformed Gramians are diagonal and equal, i.\,e.,
\begin{equation*}
T^{-1} \cP T^{-\T} = T^{\T} \cQ T = \operatorname{diag}(\sigma_1,\,\ldots,\,\sigma_n), \quad \sigma_1 \ge \sigma_2 \ge \ldots \ge \sigma_n \ge 0.
\end{equation*}
The matrix $T$ can be obtained from the Gramians $\cP$ and $\cQ$ whose computation or approximation is the most expensive part of the method. More precisely, $\cP$ and $\cQ$ are the unique solutions of the \emph{Lyapunov equations}
\begin{equation*}
 A \cP  + \cP A^\T + BB^\T = 0, \quad A^\T \cQ + \cQ A + C^\T C = 0.
\end{equation*}
If the system is controllable and observable, then $\cP > 0$ and $\cQ > 0$. Then with the Cholesky factorizations $\cP = \cR \cR^\T$ and $\cQ = \cL \cL^\T$ and the singular value decomposition $\cL^\T \cR = U \Sigma V^\T$ we obtain the transformation matrix $T = \cR V\Sigma^{-\frac{1}{2}}$ (with $T^{-1} = \Sigma^{-\frac{1}{2}}U^\T \cL^\T =: W^\T$).
Partitioning $T = \begin{bmatrix} T_r, \, \widetilde{T} \end{bmatrix}$ and $W^\T = \begin{bmatrix} W_r,\, \widetilde{W} \end{bmatrix}^\T$, where $T_r,\,W_r \in \R^{\nx \times r}$, we finally obtain the ROM by $A_r = W_r^\T A T_r$, $B_r = W_r^\T B$, $C_r = CT_r$, and $D_r = D$. If $\sigma_r > \sigma_{r+1}$, then the ROM constructed in this way is asymptotically stable and 
\begin{equation*}
\left\| G -G_r\right\|_{\mathcal{H}_\infty} \le 2\sum_{j=r+1}^{\nx} \sigma_j.
\end{equation*}

On the other hand, IRKA is based on interpolating the transfer function of the FOM and optimizing the interpolation points such that a locally $\mathcal{H}_2$ optimal ROM is obtained. The ROM is again obtained by projection, i.\,e., it is given by $A_r = W_r^\T A T_r$, $B_r = W_r^\T B$, $C_r = CT_r$, and $D_r = D$ for some appropriately chosen projection matrices $T_r,\,W_r \in \R^{\nx \times r}$ with $W_r^\T T_r = I_r$. The algorithm makes use of the following fact: If $\{s_1,\,\ldots,\,s_r\} \subset \C$ are chosen such that $s_iI_{n_x}-A$ and $s_iI_{r}-A_r$ are invertible for $i=1,\,\ldots,\,r$ and $\{b_1,\,\ldots,\,b_r\} \subset \C^{\enu}$ and $\{c_1,\,\ldots,\,c_r\} \subset \C^{\ny}$ are given right and left \emph{tangential directions}, then
\begin{align*}
    (s_i I_{\nx}-A)^{-1}Bb_i \in \range(T_r), \quad \big(c_i^\H C(s_i I_{\nx}-A)^{-1}\big)^{\H} \in \range(W_r)
\end{align*}
(where $\range( \cdot )$ denotes the image of its matrix argument) implies that 
\begin{equation}\label{eq:IRKAcond}
  G(s_i) b_i = G_r(s_i) b_i, \quad c_i^\H G(s_i) = c_i^\H G_r(s_i), \quad c_i^\H G'(s_i) b_i = c_i^\H G_r'(s_i) b_i.
\end{equation}
It has been shown in \cite{bunse-gerstner2010, GugercinH22008, van_dooren_h2-optimal_2008} that interpolation points and tangential directions that lead to $\htwo$ optimal ROMs satisfy the first-order necessary optimality conditions (which are tangential Hermite interpolation conditions)
\begin{equation} \label{eq:IRKAopt}
  G(-\lambda_i)\hat{b}_i = G_r(-\lambda_i)\hat{b}_i, \quad \hat{c}_i^\H G(-\lambda_i) = \hat{c}_i^\H G_r(-\lambda_i),\quad \hat{c}_i^\H G'(-\lambda_i) \hat{b}_i = \hat{c}_i^\H G_r'(-\lambda_i) \hat{b}_i, \quad i=1,\,\ldots,\,r.
\end{equation}
Here, $\{\lambda_1,\,\ldots,\,\lambda_r\}$ are the eigenvalues of $A_r$ and $\hat{c}_i := C_rx_i$ and $\hat{b}_i := B_r^\T y_i$, where $\{x_1,\,\ldots,\,x_r\}$ and $\{y_1,\,\ldots,\,y_r\}$ are the corresponding right and left eigenvectors with $y_i^\H x_i = 1$ for $i=1,\,\ldots,\,r$. Since these eigenvalues and eigenvectors are not known a priori, IRKA sets up a fixed point iteration in order to converge to these by alternatingly updating the ROM and computing new interpolation points and tangential directions from this ROM. For the details, we refer to \cite{GugercinH22008}. 

\subsection{Structure-Preserving Model Order Reduction Methods}\label{subsec:struct}
In this section we briefly describe some extensions of BT and IRKA to the system structures that are considered in this paper. In particular, these are the methods that we use for comparison with our new approach in Section~\ref{sec:NumExp}. We remark that, besides often being of worse approximation quality, the structure-preserving methods often lack desirable theoretical properties such as the existence of an error bound. 

BT is adapted in \cite{Polyuga2010} in several ways to preserve the pH structure. One method that is presented there is called \emph{effort constraint balanced truncation (pH-BT)}. In pH-BT, the balancing transformation $T$ from standard BT is applied to the system $\Sigma_{\text{pH}}$ to construct the balanced pH system as
\begin{align*}
  \Sigma_{\text{pH,bal}}:\,
  \begin{cases}
    \dot x_{\rm b}(t) = (J_{\rm b}-R_{\rm b})Q_{\rm b}x_{\rm b}(t)+B_{\rm b} u(t), \\
    y_{\rm b}(t)=B_{\rm b}^\T Q_{\rm b} x_{\rm b}(t),
  \end{cases}
\end{align*}
where $J_{\rm b} := T^{-1} J T^{-\T}$, $R_{\rm b} : = T^{-1} R T^{-\T}$, $Q_{\rm b} := T^{\T} Q T$, and $B_{\rm b} = T^{-1} B$. Note that this transformation preserves the pH structure. By partitioning the transformed matrices as
\begin{equation*}
 J_{\rm b} = \begin{bmatrix} J_{11} & J_{12} \\ J_{21} & J_{22} \end{bmatrix}, \quad R_{\rm b} = \begin{bmatrix} R_{11} & R_{12} \\ R_{21} & R_{22} \end{bmatrix}, \quad Q_{\rm b} = \begin{bmatrix} Q_{11} & Q_{12} \\ Q_{21} & Q_{22} \end{bmatrix}, \quad B_{\rm b} = \begin{bmatrix} B_{1} \\ B_{2} \end{bmatrix}
\end{equation*}
with $J_{11},\,R_{11},\,Q_{11} \in \R^{r \times r}$, a pH ROM with state-space dimension $r$ is then defined by
\begin{align*}
\Sigma_{\text{pH,r}}:\;
  \begin{cases}
    \dot x_r(t) = (J_{11}-R_{11})\big(Q_{11}-Q_{12}Q_{22}^{-1}Q_{21}\big)x_r(t) +B_1 u(t), \\
         y_r(t) = B_1^\T \big(Q_{11}-Q_{12}Q_{22}^{-1}Q_{21}\big) x_r(t).
  \end{cases}
\end{align*}

The adaption of IRKA to pH systems, called \emph{pH-IRKA}, is studied in \cite{Gugercin2012}. There the construction of a pH interpolant is achieved by using the following observation: If we construct the right projection matrix $T_r \in \R^{\nx \times r}$ such that for a set of interpolation points $\{s_1,\,\ldots,\,s_r\} \subset \C$ (assumed to be closed under complex conjugation) and tangential directions $\{b_1,\,\ldots,\,b_r\} \subset \C^{\enu}$ we have that
\begin{equation*}
  (s_i I_{\nx}-(J-R)Q)^{-1}B b_i \in \range(T_r), \quad i=1,\,\ldots,\,r
\end{equation*}
and, moreover, choose the left projection matrix as $W_r := Q T_r\big(T_r^\T Q T_r\big)^{-1}$, then a pH ROM is given by
\begin{align*}
  \Sigma_{\text{pH,r}}:
  \begin{cases}
    \dot x_r(t)=(J_r-R_r)Q_rx_r(t)+B_r u(t), \\
         y_r(t)=B_r^{\T}Q_r x_r(t),
  \end{cases}
\end{align*}
where $J_r = W_r^\T J W_r$, $Q_r=T_r^\T Q T_r$, $R_r=W_r^\T R W_r$, and $B_r=W_r^\T B$. This ROM does, in general, only fulfill the first of the tangential interpolation conditions in~\eqref{eq:IRKAcond}, but a fixed-point iteration with the flavor of standard IRKA can be set up to obtain a ROM that fulfills at least the first of the $\mathcal{H}_2$ optimality conditions in \eqref{eq:IRKAopt}.

A modification of BT to SSO systems, called \emph{second-order balanced truncation (SO-BT)}, is discussed in \cite{Reis2008}. There, the original SSO system is rewritten as a symmetric first-order system of the form
\begin{align*}
  \Sigma_{\text{fo}}:\;
  \begin{cases}
    \begin{bmatrix}
    I_{\nx} & 0 \\ 0 & M
  \end{bmatrix} \dot q(t) = \begin{bmatrix}
    0 & I_{\nx} \\ -K & -D
  \end{bmatrix} q(t) + \begin{bmatrix}
    0 \\ B
  \end{bmatrix} u(t), \\
    \hphantom{\begin{bmatrix}
    I_{\nx} & 0 \\ 0 & M
  \end{bmatrix}}y(t) = \begin{bmatrix} B^\T & 0 \end{bmatrix} q(t).
  \end{cases}
\end{align*}
Under the condition that $K$ and $M$ are invertible, $\Sigma_{\text{fo}}$ is a system of ordinary differential equations and if it is asymptotically stable, we can apply standard BT. The corresponding controllability and observability Gramians can be partitioned according to the block structure of the system, i.\,e.,
\begin{equation*}
 \cP = \begin{bmatrix} \cP_{\rm p} & \cP_{12} \\ \cP_{12}^\T & \cP_{\rm v} \end{bmatrix}, \quad \cQ = \begin{bmatrix} \cQ_{\rm p} & \cQ_{12} \\ \cQ_{12}^\T & \cQ_{\rm v} \end{bmatrix}. 
\end{equation*}
In \cite{Reis2008} it has been shown that in our case, $\cP_{\rm p} = \cQ_{\rm v}$. Moreover, if the system $\Sigma_{\text{fo}}$ is controllable and observable, then there exists a Cholesky factorization $\cP_{\rm p} = \cR_{\rm p} \cR_{\rm p}^\T$ with an  invertible matrix $\cR_{\rm p} \in \R^{\nx \times \nx}$ and we can compute the singular value decomposition $\cR_{\rm p}^\T M \cR_{\rm p} = U_{\rm p} \Sigma U_{\rm p}^\T$. Then with the transformation matrix $T:= \cR_{\rm p} U_{\rm p} \Sigma^{-\frac{1}{2}}$ we obtain the \emph{position/velocity balanced system}
\begin{align*}
  \Sigma_{\text{sso,bal}}:\;
  \begin{cases}
    T^\T M T \ddot x_{\rm b}(t) + T^\T D T \dot x_{\rm b}(t)+ T^\T K T x_{\rm b}(t) = T^\T B u(t),\\
    \phantom{T^\T M T \ddot x_{\rm b}(t) + T^\T D T \dot x_{\rm b}(t)+ T^\T K T }y_{\rm b}(t) = B^{\T} T x_{\rm b}(t),
  \end{cases}
\end{align*}
in which the correspondingly transformed position/velocity Gramians $\cP_{\rm p,bal}$ and $\cQ_{\rm v,bal}$ are diagonal and equal. By partitioning $T = \begin{bmatrix} T_r, \, \widetilde{T} \end{bmatrix}$, we finally obtain a SSO ROM as in \eqref{eq:mechrom} with the desired symmetries by setting $M_r = T_r^\T M T_r$, $K_r = T_r^\T K T_r$, $D_r = T_r^\T D T_r$, and $B_r = T_r^\T B$. 

Note that an adaption of IRKA to second-order systems is still an open problem, but some advances in the direction of optimality conditions are available in \cite{BeaB14}. 

\section{Our Method}%
\label{sec:our_method}
Our method for computing low-order pH or SSO realizations that approximate the input/output map of a given system is based on making an ansatz for a parametrized low-order pH or SSO model, respectively, and then directly optimizing the model parameters. The objective functional that is minimized is based on the difference of the transfer functions of the FOM and ROM evaluated on the imaginary axis. 

Therefore, we first present a parametrization of pH and SSO systems. 
Then we discuss our proposed objective functional and explain our optimization 
procedure. This optimization is gradient-based, so we further derive the 
gradients that are involved in this process.

\subsection{Structured Parametrized Transfer Functions}

Throughout the next sections, we will rely on the following functions to manage 
the construction of the system matrices from parameter vectors.
\begin{definition}[Reshaping operations]
  \begin{enumerate}[a)] 
    \item The function family  
      \begin{align*}
        \vtf_m: \C^{n \cdot m} \rightarrow \C^{n\times m}, \quad  v \mapsto
        \begin{bmatrix}
          v_1 & v_{n+1} & \dots & v_{m(n-1)+1}\\
          v_2 & v_{n+2} & \dots & v_{m(n-1)+2}\\
          \vdots & \vdots &  & \vdots \\
          v_n & v_{2n} & \dots & v_{nm}
        \end{bmatrix}
      \end{align*}
      reshapes a vector into an accordingly sized matrix with $m$ columns. Note that $\vtf$ can be read as vector-to-full (matrix).
      Its inverse is given by
      \begin{align*}
        \ftv: \C^{n \times m}  \rightarrow \C^{n \cdot m}, \quad A \mapsto
        \begin{bmatrix}
          a_{1,1} &a_{2,1} &\dots & a_{n,1} & a_{1,2} & \dots & a_{n,m}
        \end{bmatrix}^\T,
      \end{align*}
      where $\ftv$ stands for full (matrix)-to-vector. The latter is the standard vectorization operator which is commonly denoted by $\operatorname{vec}$.
    \item The function 
      \begin{align*}
        \vtu : \C^{n(n+1)/2}  \rightarrow \C^{n\times n}, \quad v \mapsto
        \begin{bmatrix}
          v_1 & v_2 & \dots & v_n \\
          0 & v_{n+1}& \dots & v_{2n-1} \\
          \vdots & \vdots & \ddots & \vdots \\
          0 & 0 & \dots & v_{n(n+1)/2} \\
        \end{bmatrix}
      \end{align*}
      maps a vector of length $n(n+1)/2$ to an $n\times n$ upper triangular 
      matrix (where $\vtu$ stands for vector-to-upper (triangular)), while the function 
      \begin{align*}
        \utv : \C^{n\times n}  \rightarrow \C^{n(n+1)/2}, \quad A \mapsto
        \begin{bmatrix}
          a_{1,1} & a_{1,2} & \dots & a_{1,n} & a_{2,2} & \dots & a_{n,n}
        \end{bmatrix}^\T
      \end{align*}
      maps the upper triangular part of an $n\times n$ matrix to a vector (where $\utv$ stands for upper (triangular)-to-vector).
    \item The function
      \begin{align*}
        \vtsu : \C^{n(n-1)/2}  \rightarrow \C^{n\times n}, \quad v \mapsto
        \begin{bmatrix}
          0 & v_1 & v_2 & \dots  & v_{n-1} \\
          0 & 0   & v_n & \dots  & v_{2n-3}  \\
          \vdots & \vdots  & \vdots   & \ddots & \vdots \\
          0 & 0   & 0   &  \dots     & v_{n(n-1)/2} \\
          0 & 0   & 0   & \dots      & 0\\
        \end{bmatrix}
      \end{align*}
      maps a vector of length $n(n-1)/2$ to an $n\times n$ strictly upper 
      triangular matrix (where $\vtsu$ stands for vector-to-strictly upper (triangular)), while the function 
      \begin{align*}
        \sutv : \C^{n \times n}  \rightarrow \C^{n(n-1)/2}, \quad A \mapsto
        \begin{bmatrix}
          a_{1,2} & a_{1,3} & \dots & a_{1,n} & a_{2,3} & \dots & a_{n-1,n}
        \end{bmatrix}^\T
      \end{align*}
       maps the strictly upper triangular part of an $n\times n$ matrix to a 
       vector (where $\sutv$ means strictly upper (triangular)-to-vector).
  \end{enumerate}
\end{definition}

Using these reshaping operations, we can define a parametrization of pH and SSO systems as follows.

\subsubsection{Parametrized Port-Hamiltonian Systems}
We provide a parametrization of a pH system that automatically satisfies the given structural constraints on the system matrices. The following lemma is a direct consequence of this structure.

\begin{lemma}
  Let $\theta \in \R^{n_\theta}$ be a parameter vector with $n_\theta=n_x\left(\frac{3n_x+1}{2}+n_u\right)$. 
  Furthermore, let $\theta$ be partitioned as  $\theta:=\begin{bmatrix}\theta_J^\T,\,\theta_R^\T,\,\theta_Q^\T,\,\theta_B^\T\end{bmatrix}^\T$ with $\theta_J\in\R^{n_x(n_x-1)/2}$, $\theta_R,\,\theta_Q\in\R^{n_x(n_x+1)/2}$, and $\theta_B\in\R^{n_x \cdot n_u}$. Further define the matrices
  \begin{subequations}
  \begin{align}
    \pJ(\theta) &=\vtsu(\theta_J)^\T-\vtsu(\theta_J),\label{eq:pJ}\\
    \pR(\theta) &=\vtu(\theta_R)^\T  \vtu(\theta_R),\\
    \pQ(\theta) &=\vtu(\theta_Q)^\T  \vtu(\theta_Q),\label{eq:pQ}\\
    \pB(\theta) &=\vtf_{n_u}(\theta_B).
  \end{align}
  \end{subequations}
  Then, to each $\theta \in \R^{n_\theta}$ one can assign the pH system
  \begin{align} \label{eq:pHdelta}
    \pHsys(\theta): \;
    \begin{cases}
      \dot x(t) = \left( \pJ(\theta)-\pR(\theta) \right)\pQ(\theta) x(t)+ \pB(\theta) u(t),\\
      y(t) = \pB(\theta)^\T \pQ(\theta) x(t).
    \end{cases}
  \end{align}
  Conversely, to each pH system $\pHsys$ with $\nx$ states and $n_u$ inputs and outputs one can assign a vector $\theta \in \R^{n_\theta}$ such that $\pHsys = \pHsys(\theta)$ with $\pHsys(\theta)$ as in \eqref{eq:pHdelta}.
\end{lemma}


In the next result, we provide the gradients of the singular values of the difference between a given transfer function $G \in \mathcal{RH}_\infty^{n_u \times n_u}$ (of a possibly large-scale system) and the transfer function of the pH system~\eqref{eq:pHdelta}, given by
\begin{equation} \label{eq:phtf}
\pHtf(s,\theta)=\pB(\theta)^\T \pQ(\theta)\left(sI_{n_x}-\left(\pJ(\theta)-\pR(\theta)\right)\pQ(\theta)\right)^{-1}\pB(\theta),
\end{equation}
evaluated at some given value $s_0 \in \overline{\C^+} := \{ \lambda \in \C \;|\; \Real(\lambda) \ge 0 \}$ with respect to the parameter vector~$\theta$.

\begin{theorem} 
  \label{thm::pHgradients}
  Let $\theta_0 \in \R^{n_\theta}$ be given and assume that $G \in \mathcal{RH}_\infty^{n_u \times n_u}$ and $\pHtf(\cdot,\theta_0) \in \rhinf^{n_u \times n_u}$ as in \eqref{eq:phtf}. Suppose further that for given $s_0 \in \overline{\C^+}$, the $j$-th singular value of $G(s_0)-\pHtf(s_0,\theta_0)$ is nonzero and simple and let $\uu \in \C^{n_u}$ and $\vv \in \C^{n_u}$ be the corresponding left and right singular vectors, respectively.
  
  Then the function $\theta \mapsto \sigma_j(G(s_0)-\pHtf(s_0,\theta))$, where $\sigma_j$ denotes the $j$-th singular value, is differentiable in a neighborhood of $\theta_0$. Moreover, we define the short-hand notations $J_0:=J(\theta_0),\,R_0:=R(\theta_0),\,Q_0:=Q(\theta_0),\,B_0:=B(\theta_0)$ and suppose that
  \begin{align*}
    \dyn&:=s_0 I_{n_x} -\left(\pJ_0-\pR_0\right)\pQ_0
  \end{align*}
  is invertible. With
  \begin{align*}
    \dJ &:=-\pQ_0 \dyn^{-1}  \inp_0 \vv \uu^\H \pB_0^{\T} \pQ_0  \dyn^{-1}, \\
    \dQ &:=\dyn^{-1} \inp_0  \vv  \left(\uu^\H \pB_0^\T+ \uu^{\H}  \pB_0^{\T} \pQ_0  \dyn^{-1}  \left(\pJ_0-\pR_0\right)\right),
  \end{align*}
  we obtain the gradient
  \begin{align*}
    \nabla_\theta \sigma_j(G(s_0)-\pHtf(s_0,\theta_0)):=\begin{bmatrix}
    \dta^\T,\, \dtb^\T,\, \dtc^\T,\, \dtd^\T
    \end{bmatrix}^\T,
  \end{align*}
  where
  \begin{subequations}
      \begin{align}
        \dta &= \Real\left( \sutv(-\dJ+\dJ^\T)\right),\label{eq::dta}\\
        \dtb &= \Real\left( \utv \left(\vtu(\theta_R) \dJ+\vtu(\theta_R) \dJ^\T\right)\right),\label{eq::dtb}\\
        \dtc &= \Real\left( \utv \left(\vtu(\theta_Q) \dQ+\vtu(\theta_Q) \dQ^\T\right)\right),\label{eq::dtc}\\
        \dtd &= \Real\left( \ftv \left( (\vv \uu^\H  \pB_0^{\T} \pQ_0 \dyn^{-1})^\T+\pQ_0\dyn^{-1}\inp_0 \vv \uu^\H \right)\right).\label{eq::dtd}
      \end{align}
  \end{subequations}
\end{theorem}

Before we can prove Theorem~\ref{thm::pHgradients}, we need the following lemma concerning the reshaping functions.
\begin{lemma}\label{lem::trace}
  Let $A \in \C^{m \times n}$ and let $e^{(j)}_i \in \C^j$ denote the $i$-th standard basis vector of $\C^{j}$. Then
    \begin{align}
      \tr\left(A\vtf_m\big(e_i^{(nm)}\big)\right)=\big(e_i^{(nm)}\big)^\T \ftv\big(A^\T\big)\label{eq::tracevtf}.
    \end{align}
    Now let $m=n$ and define $n_1 := \frac{n(n+1)}{2}$ and $n_2 := \frac{n(n-1)}{2}$. Then
    \begin{align}
      \tr\left(A\vtu\big(e_i^{(n_1)}\big)\right) = \big(e_i^{(n_1)}\big)^\T \utv\big(A^\T\big), \label{eq::tracevtu}\\
      \tr\left(A\vtsu\big(e_i^{(n_2)}\big)\right) = \big(e_i^{(n_2)}\big)^\T \sutv\big(A^\T\big)\label{eq::tracevtsu}.
    \end{align}
\end{lemma}

\begin{proof}
  Note that $\vtf_m\big(e_i^{(nm)}\big) \in \R^{n \times m}$ can be expressed as $e_k^{(n)}\big(e_\ell^{(m)}\big)^\T$, where $i = (\ell-1)n+k$. Since the trace is invariant under cyclic permutations\footnote{For arbitrary $X\in \C^{n \times m}, Y \in \C^{m \times p},$ and $Z\in \C^{p \times n}$ we have that $\tr(XYZ)=\tr(ZXY)=\tr(YZX)$.}, we have
  \begin{align*}
    \tr\left(A\vtu\big(e_i^{(nm)}\big)\right) = \tr\left(Ae_k^{(n)}\big(e_\ell^{(m)}\big)^\T\right) = \tr\left(\big(e_\ell^{(m)}\big)^\T A e_k^{(n)}\right) = a_{\ell,k}.
  \end{align*}
  On the other hand, it can be verified from the definition of $\ftv$ that  $\big(e_i^{(nm)}\big)^\T \ftv(A^\T) = a_{\ell,k}$.
  The equalities \eqref{eq::tracevtu} and \eqref{eq::tracevtsu} can be shown analogously by identifying the nonzero element of $\vtu\big(e_i^{(n_1)}\big)$ and $\vtsu\big(e_i^{(n_2)}\big)$, respectively, and using the cyclic permutation invariance of the trace.
\end{proof}

We proceed with the proof of Theorem~\ref{thm::pHgradients}.
\begin{proof}
  We show differentiability with respect to $\theta_J$ and accordingly the partial gradient \eqref{eq::dta}.
  Fix an $i\in \{1,\,\dots,\,n_x(n_x-1)/2\}$ and let $e_i$  be the $i$-th standard basis vector of $\R^{n_x(n_x-1)/2}$. Let $\dd\theta_i(\epsilon) := \begin{bmatrix} \epsilon e_i^\T,\,0 \end{bmatrix}^\T \in \R^{n_\theta}$ and $\Delta_i := \pJ(e_i)$. 
  
  We have 
  \begin{align*}
   \pHtf(s_0,\theta_0 + \dd\theta_i(\varepsilon)) &= B_0^\T Q_0(s_0 I_{n_x} - (J_0 + \varepsilon \Delta_i - R_0)Q_0)^{-1}B_0 \\
   &= -B_0^\T Q_0( \varepsilon \Delta_i Q_0 - \dyn)^{-1}B_0.
  \end{align*}
  Thus, $\varepsilon \mapsto \pHtf(s_0,\theta_0 + \dd\theta_i(\varepsilon))$ is a rational matrix-valued function that admits a Taylor series expansion at zero which, according to \cite{MehS05}, is
  \begin{equation}\label{eq:taylor}
    \pHtf(s_0,\theta_0 + \dd\theta_i(\varepsilon)) = \pHtf(s_0,\theta_0) + \varepsilon B_0^\T Q_0 \dyn^{-1} \Delta_i Q_0 \dyn^{-1}B_0 + \mathcal{O}(\varepsilon^2).
  \end{equation}
  Due to the nonzero and simplicity assumptions for the $j$-th singular value of $G(s_0) - \pHtf(s_0,\theta_0)$, by \cite{Lan64} the function $\varepsilon \mapsto \sigma_j\left(G(s_0) - \pHtf(s_0,\theta_0 + \dd\theta_i(\varepsilon))\right)$ is differentiable at zero. Together with \eqref{eq:taylor} we obtain
  \begin{align*}
    \frac{\dd}{\dd \varepsilon} \sigma_j\left(G(s_0) - \pHtf(s_0,\theta_0 + \dd\theta_i(\varepsilon))\right) \Big|_{\varepsilon = 0} &= -\Real\left(\uu^\H B_0^\T Q_0 \dyn^{-1} \Delta_i Q_0 \dyn^{-1}B_0 \vv\right) \\ &= -\Real\left(\tr\left(Q_0 \dyn^{-1}B_0 \vv\uu^\H B_0^\T Q_0 \dyn^{-1} \Delta_i\right)\right) \\ &= \Real(\tr(Y_1\Delta_i)).
  \end{align*}
  Since $\Delta_i = \vtsu(e_i)^\T - \vtsu(e_i)$ we have
  \begin{align*}
    \Real\left(\tr(\dJ \Delta_i)\right) &= \Real\left(\tr\left(\dJ (\vtsu(e_i)^\T-\vtsu(e_i))\right)\right) \\ 
                    &= \Real\left(\tr\left(\dJ \vtsu(e_i)^\T)-\tr(\dJ\vtsu(e_i)\right)\right) \\
                    &= \Real\left(\tr\left((\dJ^\T-\dJ) \vtsu(e_i)\right)\right) = (\dta)_i,
  \end{align*}
  with $\dta$ as in \eqref{eq::dta}.
  The last equality is due to Lemma~\ref{lem::trace}. Differentiability with respect to the other components of $\theta$ can be shown analogously as well as the formulas \eqref{eq::dtb}--\eqref{eq::dtd}, again by making use of Lemma~\ref{lem::trace}.
\end{proof}
 
\subsubsection{Parametrized Mechanical Systems}
Similarly to the previous parametrization of pH systems, we can proceed in the case of SSO systems.
\begin{lemma}
  Let $\theta \in \R^{n_\theta}$ be a parameter vector with $n_\theta = n_x\left(\frac{3n_x+3}{2}+n_u\right)$. 
  Furthermore, let $\theta$ be partitioned as $\theta:=\begin{bmatrix}\theta_M^\T,\,\theta_D^\T,\,\theta_K^\T,\,\theta_B^\T\end{bmatrix}^\T$ with $\theta_M,\,\theta_D,\,\theta_K\in\R^{n_x(n_x+1)/2}$ and $\theta_B\in\R^{n_x \cdot n_u}$. Further define the matrices
  \begin{subequations}
    \begin{align}
      \pM(\theta) &=\vtu(\theta_M)^\T \vtu(\theta_M),\\
      \pD(\theta) &=\vtu(\theta_D)^\T \vtu(\theta_D),\\
      \pK(\theta) &=\vtu(\theta_K)^\T \vtu(\theta_K),\\
      \pB(\theta) &=\vtf_{n_u}(\theta_B).
    \end{align}
  \end{subequations}
  Then, to each $\theta \in \R^{n_\theta}$ one can assign the second-order system
\begin{align} \label{eq:mechdelta}
  \mechSys(\theta) :\;
  \begin{cases}
    \pM(\theta) \ddot x(t) + \pD(\theta) \dot x(t) + \pK(\theta) x(t) = \pB(\theta) u(t),\\
    \hphantom{\pM(\theta) \ddot x(t) + \pD(\theta) \dot x(t) + \pK(\theta)} y(t) = \pB(\theta)^\T x(t).
  \end{cases}
\end{align}
 Conversely, to each SSO system $\mechSys$ with $\nx$ states and $n_u$ inputs and outputs one can assign a vector $\theta \in \R^{n_\theta}$ such that $\mechSys = \mechSys(\theta)$ with $\mechSys(\theta)$ as in \eqref{eq:mechdelta}.
 \label{lem:sndparam}
\end{lemma}

Next, we formulate the analogue of Theorem~\ref{thm::pHgradients}. Again we provide the gradients of the norm of the difference between a given transfer function $G \in \mathcal{RH}_\infty^{n_u \times n_u}$ (of a possibly large-scale system) and the transfer function of the SSO system \eqref{eq:mechdelta}, given by
\begin{equation} \label{eq:mechtf}
\mechtf(s,\theta)=\pB(\theta)^\T(s^2 \pM(\theta) + s \pD(\theta)+ \pK(\theta))^{-1}\pB(\theta),
\end{equation}
evaluated at some given value $s_0 \in \overline{\C^+}$ with respect to the parameter vector~$\theta$.

\begin{theorem}\label{thm:mechGradients}
  Let $\theta_0 \in \R^{n_\theta}$ be given and assume that $G \in \mathcal{RH}_\infty^{n_u \times n_u}$ and $\mechtf(\cdot,\theta_0) \in \rhinf^{n_u \times n_u}$ as in \eqref{eq:mechtf}. Suppose further that for given $s_0 \in \overline{\C^+}$, the $j$-th singular value of $G(s_0)-\mechtf(s_0,\theta_0)$ is nonzero and simple and let $\uu \in \C^{n_u}$ and $\vv \in \C^{n_u}$ be the corresponding left and right singular vectors, respectively.

 Then the function $\theta \mapsto \sigma_j(G(s_0)-\mechtf(s_0,\theta))$ is differentiable in a neighborhood of $\theta_0$. Define the short-hand notations $M_0:=M(\theta_0),\,D_0:=D(\theta_0),\,K_0:=K(\theta_0),\,B_0:=B(\theta_0)$, and assume further that
 \begin{align*}
    \dynmech :=s_0^2 M_0 + s_0D_0 + K_0
  \end{align*}
  is invertible. Then with
  \begin{align*}
    Y := \dynmech^{-1}B_0\vv \uu^{\H} B_0^\T \dynmech^{-1},
  \end{align*}
  we obtain
  \begin{align*}
    \nabla_\theta \sigma_j(G(s_0)-\mechtf(s_0,\theta_0)) := \begin{bmatrix}
    \dtM^\T,\, \dtD^\T,\, \dtK^\T, \, \dtB^\T \end{bmatrix}^\T,
  \end{align*}
  where
  \begin{subequations}
    \begin{align}
      \dtM &= \Real\left(s_0^2 \cdot \utv(\vtu(\theta_M) Y+\vtu(\theta_M) Y^\T)\right),\label{eq::dtmech}\\
      \dtD &= \Real\left(s_0 \cdot \utv(\vtu(\theta_D) Y+\vtu(\theta_D) Y^\T)\right),\\
      \dtK &= \Real\left(\utv(\vtu(\theta_K) Y+\vtu(\theta_K) Y^\T)\right),\\
      \dtB &= \Real\left(\ftv(\dynmech^{-1}B_0 (\vv \uu^\H+ \uu \vv^\H))\right).
    \end{align}
  \end{subequations}
\end{theorem}

\begin{proof}
  The proof is analogous to the proof of Theorem~\ref{thm::pHgradients}. We again show differentiability with respect to $\theta_M$ and the corresponding partial gradient \eqref{eq::dtmech}. Fix $i \in \{ 1, \dots, n_x(n_x+1)/2 \}$ and let $\dd \theta_i(\epsilon):= \begin{bmatrix}
    \epsilon e_i^\T & 0
  \end{bmatrix}^T \in \R^{n_\theta}$ where $e_i$ is the $i$-th standard unit vector in $\R^{n_x(n_x+1)/2}$. Note that $\epsilon  \mapsto M\left(\theta_0+\dd \theta_i(\epsilon)\right)$ admits a Taylor series expansion at 0. It is given by
  \begin{align*}
    M(\theta_0 + \dd \theta_i(\epsilon)) &= \vtu\left(\theta_0+\dd \theta_i(\epsilon) \right)^\T \vtu \left( \theta_0+ \dd \theta_i(\epsilon)\right) 
                                         = M_0 + \epsilon\Delta_i + \mathcal{O}(\epsilon^2),
  \end{align*}
  where $\Delta_i= \vtu(\theta_M)^\T \vtu(e_i) + \vtu(e_i)^\T \vtu(\theta_M)$. Therefore, the Taylor series expansion of $\epsilon  \mapsto \mechtf(s_0, \theta_0+\dd \theta_i(\epsilon))$ at zero is given by
  \begin{align}
    \mechtf(s_0,\theta_0+\dd \theta_i(\epsilon))&=B_0^\T\left(\epsilon s_0^2\Delta_i + \dyn +\mathcal{O}(\epsilon^2)\right)^{-1}B_0 \nonumber \\
    &=\mechtf(s_0, \theta_0) - \epsilon s_0^2 B_0^\T \dyn^{-1} \Delta_i \dyn^{-1} B_0 + \mathcal{O}(\epsilon^2). \label{eq:taylormech}
  \end{align}
  For the derivative of $\epsilon \mapsto \sigma_j(G(s_0)-\mechtf(s_0, \theta_0+\dd \theta_i(\epsilon)))$ at zero, from \eqref{eq:taylormech}, we obtain
  \begin{align*}
    \frac{\dd}{\dd \epsilon} \sigma_j( G(s_0)-\mechtf(s_0, \theta_0+\dd \theta_i(\epsilon))) \Big|_{\epsilon=0} &= \Real (s_0^2 \uu^\H B_0^\T \dyn^{-1} \Delta_i \dyn^{-1} B_0 \vv) \\
                                                                                                            &= \Real \big(s_0^2 \tr (\dyn^{-1}B_0 \vv \uu^\H B_0^\T \dyn^{-1} \Delta_i)\big) \\
                                                                                                            &= \Real \big(s_0^2 \tr (Y \Delta_i)\big).
  \end{align*}
  Since $\Delta_i = \vtu(\theta_M)^\T \vtu(e_i) + \vtu(e_i)^\T \vtu(\theta_M)$, we have that
  \begin{align*}
    \Real\left(s_0^2 \tr (Y \Delta_i)\right) &= \Real\left(s_0^2 \tr(Y \vtu(\theta_M)^\T \vtu(e_i))+ \tr(Y \vtu(e_i)^\T \vtu(\theta_M))\right) \\
                                             &= \Real \left( s_0^2 (e_i^\T \utv(\vtu(\theta_M) Y^\T) + e_i^\T \utv(\vtu(\theta_M)Y)) \right) = (\dtM)_i.
  \end{align*}
  The other parts of the gradient can be shown analogously.
\end{proof}

\subsection{Optimization of the ROM Parameters}

In the following, we will explain how we optimize the parameters of the ROM to minimize the distance of its transfer function to the transfer function of a given model with respect to the $\hinf$ norm. Throughout the following section, we denote the transfer function of a given (large-scale) model by $G$ and some parametrized transfer function approximation by $G_r(\cdot, \theta)$. We assume that $G \in \mathcal{RH}_\infty^{n_u \times n_u}$ and will construct a sequence $(\theta_i)_{i \in \mathbb{N}}$ such that $G_r(\cdot,\theta_i) \in \rhinf^{n_u \times n_u}$ for $i=0,\,1,\,2,\,\ldots$ with the goal to achieve smaller errors for increasing values of $i$. This sequence should converge to a final value $\theta_{\rm fin} \in \R^{n_\theta}$ leading to a small $\hinf$ error.

We will first explain the problems associated with directly minimizing ${\|G- G_r(\cdot,\theta)\|}_{\hinf}$ and after that detail our approach.

\subsubsection{Problems with the Direct $\hinf$ Error Minimization Approach}

Methods for computing the $\hinf$ norm of rational transfer functions exist since the early 1990s \cite{Boyd1990, Bruinsma1990}. However, these are typically not used for computing the $\hinf$ norm of $G- G_r( \cdot, \theta)$, since they require multiple solutions of eigenvalue problems of size around twice the state-space dimension of this error transfer function which has a slightly larger state-space dimension than the state-space realization of $\Sigma$ of $G$ and do not exploit the sparsity of the system matrices. Since in the MOR context, $\Sigma$ is assumed to be large and sparse, this is computationally prohibitive.

Recent developments have led to faster $\hinf$ norm computation methods for transfer functions of systems with a large state-space dimension \cite{Guglielmi2013, Voigt2015, Freitag2014, AliBMSV17a, SchV18, MitchOv2015}. However, with these methods, typically only a local maximizer of $\omega \mapsto {\|G(\mathrm{i}\omega)\|}_2$ for $\omega \in \R$ is computed. This problem is partly addressed in \cite{SchMV20}, in which a global certificate for the $\hinf$ norm computation for large-scale systems is established and which can exploit sparsity in the systems matrices. However, this is still computationally expensive compared to the $\hinf$ norm computation in, e.\,g., \cite{AliBMSV17a}.

These fast $\hinf$ norm computation methods in principle allow a direct optimization of the parameters of $G_r( \cdot, \theta)$ to minimize its distance to a given $G$ with respect to the $\hinf$ norm. However, the problem
\begin{align}
  \label{eq:hinferror18}
  \min_{\theta \in \R^{n_\theta}} \left\|G-G_r( \cdot, \theta)\right\|_\hinf
\end{align}
is a nonlinear, nonsmooth, and nonconvex optimization problem. 
While there exist some methods that solve nonsmooth optimization problems, e.\,g., \cite{BurLO05,Rus06,CurO12,Curtis2017}, these tend to converge only slowly requiring a high number of $\hinf$ norm evaluations. Thus, certifying all computed $\hinf$ norms is computationally prohibitive such that the computation of the correct $\hinf$ norm at every iterate cannot be guaranteed.
Therefore, we propose an alternative optimization approach in the next section.
In Section~\ref{sec:NumExp} we compare our alternative method with the direct $\hinf$ error minimization approach. For computing the $\hinf$ norm in this approach we use both the large-scale method in \cite{SchV18} (\texttt{linorm\_subsp}), which can exploit sparsity, and the method in \cite{BenSV12} (\texttt{ab13hd}), which uses full and dense matrices. For the direct $\hinf$ error minimization we use \texttt{GRANSO}\footnote{available at \url{https://gitlab.com/timmitchell/GRANSO/}}, which is presented in \cite{Curtis2017}.


\subsubsection{Leveled Least Squares}
Instead of minimizing $\|G-G_r( \cdot, \theta)\|_{\hinf}$ directly, we propose to minimize an alternative objective functional. Let $\gamma > 0$ be a parameter and $S := \{ s_1,\,\ldots ,\, s_k\} \subset \ri \R$ be a set of sample points. We propose to minimize
\begin{align}
  \loss(\gamma,G,G_r( \cdot,\theta),S) := \frac{1}{\gamma}\sum\limits_{s_i\in S}
  \left(\sum\limits_{j=1}^{n_u}\left(\left[\sigma_j \left(G(s_i)-G_r(s_i,\theta)\right)-\gamma\right]_+\right)^2\right)
  \label{eq:loss}
\end{align}
with respect to $\theta$, where 
\begin{align*}
  [ \cdot ]_+:  \R \rightarrow \overline{\R^+}, \quad x \mapsto 
  \begin{cases}
    x & \text{if } x\ge 0,\\
    0 & \text{if } x<0
  \end{cases}
\end{align*}
for decreasing values of $\gamma > 0$. 

In the following proposition, we summarize the characteristics of $\loss$ that motivate its utilization in $\hinf$ error minimization.
\begin{proposition}[Properties of $\loss$] Let $S =\{s_1,\,\ldots,\,s_k \} \subset \mathrm{i}\R$ and $\gamma > 0$ be fixed and let $L$ be given as in \eqref{eq:loss}. For $i=1,\,\ldots,\,k$ and $j = 1,\,\ldots,\,\enu$ define
\begin{equation*}
  f_{ij}(\theta) := \sigma_j(G(s_i)-G_r(s_i,\theta))
\end{equation*}
and $g(\theta) := \loss(\gamma,G,G_r( \cdot,\theta),S)$. Then the following statements are satisfied:
  \begin{enumerate}[i)]
    \item We have that $\loss(\gamma,G,G_r( \cdot, \theta),S)=0$ for all $\gamma>{\|G-G_r( \cdot, \theta)\|}_{\hinf}$.
    \item 
    The function $g(\cdot)$ is differentiable at $\theta_0$. Moreover, the partial derivatives of $g(\cdot)$ at $\theta_0$ are given by
    \begin{align}\label{eq:diffg}
     \begin{split}
      \frac{\partial}{\partial \theta_{\ell}}g(\theta_0) &= \frac{2}{\gamma} \sum_{f_{ij}(\theta_0) > \gamma} (f_{ij}(\theta_0) -\gamma) \frac{\partial_+}{\partial \theta_\ell} f_{ij}(\theta_0) \\ &= \frac{2}{\gamma} \sum_{f_{ij}(\theta_0) > \gamma} (f_{ij}(\theta_0) -\gamma) \frac{\partial_-}{\partial \theta_\ell} f_{ij}(\theta_0), \quad \ell=1,\,\ldots,\,n_\theta,
      \end{split}
    \end{align}
    where $\frac{\partial_+}{\partial \theta_\ell}$ and $\frac{\partial_-}{\partial \theta_\ell}$ denote the \emph{right} and \emph{left} partial derivative with respect to $\theta_\ell$.
  \end{enumerate}
  \label{prop:lossprops}
\end{proposition}

\begin{remark}\label{rem:svals}
 If the $j$-th singular value $\sigma_0 := \sigma_j(G(s_i)-G_r(s_i,\theta_0)) > 0$ is not simple, then $f_{ij}(\cdot)$ is in general not differentiable in $\theta_0$. However, it is \emph{(G\^ateaux) semi-differentiable} in $\theta_0$, i.\,e., the limit
 \begin{equation*}
  \lim_{\varepsilon \searrow 0} \frac{f_{ij}\big(\theta_0 + \varepsilon\hat{\theta}\big) - f_{ij}(\theta_0)}{\varepsilon}
 \end{equation*}
 exists for all directions $\hat{\theta} \in \R^{n_\theta}$.
 On the other hand, the function 
 $\sum_{\{ j \; | \; f_{ij}(\theta_0) = \sigma_0 \}} f_{ij}(\cdot)$ 
 is still differentiable in $\theta_0$ -- in other words, the different branches of singular value curves that intersect each other at $\theta_0$ still add up smoothly. 
 
 If all singular values of $G(s_i)-G_r(s_i,\theta_0)$ greater than $\gamma$ are simple, then in \eqref{eq:diffg} all semi-derivatives can be replaced by derivatives which can be computed by either Theorem~\ref{thm::pHgradients} or Theorem~\ref{thm:mechGradients}.
\end{remark}

We continue with the proof of Proposition~\ref{prop:lossprops}.
\begin{proof}
  Assertion i) follows directly from the definition of $L$. 
  
  It remains to show ii): For $i=1,\,\ldots,\,k$ define the index sets
  \begin{align*}
   \mathcal{I}_i^+ &:= \{ j \in \{1,\,\ldots,\,\enu\} \; | \; f_{ij}(\theta_0) > \gamma \}, \\
   \mathcal{I}_i^0 &:= \{ j \in \{1,\,\ldots,\,\enu\} \; | \; f_{ij}(\theta_0) = \gamma \}, \\
   \mathcal{I}_i^- &:= \{ j \in \{1,\,\ldots,\,\enu\} \; | \; f_{ij}(\theta_0) < \gamma \},
  \end{align*}
  as well as 
  \begin{equation*}
  g_i^+(\theta) := \sum_{j \in \cI_i^+} \big([f_{ij}(\theta)-\gamma]_+\big)^2, \quad g_i^0(\theta) := \sum_{j \in \cI_i^0} \big([f_{ij}(\theta)-\gamma]_+\big)^2, \quad g_i^-(\theta) := \sum_{j \in \cI_i^-} \big([f_{ij}(\theta)-\gamma]_+\big)^2.
  \end{equation*}
  
  \emph{Case 1:} Clearly, if $f_{ij}(\theta_0) < \gamma$, then by continuity of $f_{ij}(\cdot)$, $f_{ij}(\theta) < \gamma$ for all $\theta$ in some neighborhood $\cN(\theta_0)$ of $\theta_0$. Thus, $g_i^-(\cdot)$ is differentiable in $\cN(\theta_0)$ and the gradient is zero.
  
  \emph{Case 2:} If $f_{ij}(\theta_0) > \gamma$, then $f_{ij}(\theta) > \gamma$ and $([f_{ij}(\theta) - \gamma]_+)^2 = (f_{ij}(\theta)-\gamma)^2$ for all $\theta$ in a neighborhood $\cN(\theta_0)$ of $\theta_0$. Hence, by \cite[Thm.~3.5]{Del19}, $([f_{ij}(\cdot)-\gamma]_+)^2$ is semi-differentiable in $\cN(\theta_0)$.
  Then for $\theta \in \cN(\theta_0)$ the chain rule implies
 \begin{align*}
   \frac{\partial_+}{\partial \theta_\ell}([f_{ij}(\theta)-\gamma]_+)^2 &= 2 (f_{ij}(\theta)-\gamma) \frac{\partial_+}{\partial \theta_\ell}f_{ij}(\theta), \\ 
   \frac{\partial_-}{\partial \theta_\ell}([f_{ij}(\theta)-\gamma]_+)^2 &= 2 (f_{ij}(\theta)-\gamma) \frac{\partial_-}{\partial \theta_\ell}f_{ij}(\theta), \quad \ell=1,\,\ldots,\,n_\theta,
 \end{align*}
 By Remark~\ref{rem:svals}, for $\theta \in \cN(\theta_0)$ we have 
 \begin{equation*}
  \sum_{j \in \cI_i^+} (f_{ij}(\theta)-\gamma) \frac{\partial_+}{\partial \theta_\ell}f_{ij}(\theta) = \sum_{j \in \cI_i^+} (f_{ij}(\theta)-\gamma) \frac{\partial_-}{\partial \theta_\ell}f_{ij}(\theta). 
 \end{equation*}
 Hence, $g_i^+(\cdot)$ is partially differentiable with respect to $\theta_\ell$ in $\theta \in \cN(\theta_0)$ and we have
 \begin{equation*}
   \frac{\partial}{\partial \theta_{\ell}}g_i^+(\theta) =  2 \sum_{j \in \cI_i^+} (f_{ij}(\theta) -\gamma) \frac{\partial_+}{\partial \theta_\ell} f_{ij}(\theta) = 2 \sum_{j \in \cI_i^+} (f_{ij}(\theta) -\gamma) \frac{\partial_-}{\partial \theta_\ell} f_{ij}(\theta), \quad \ell=1,\,\ldots,\,n_\theta.
 \end{equation*}
 Due to the continuity of all partial derivatives, $g_i^+(\cdot)$ is differentiable.
 
 \emph{Case 3:} It remains to consider the cases in which $f_{ij}(\theta_0) = \gamma$. As for Case 2 we obtain
 \begin{align*}
  \frac{\partial}{\partial \theta_{\ell}}g_i^0(\theta_0) &=  2 \sum_{j \in \cI_i^0} (f_{ij}(\theta_0) -\gamma) \frac{\partial_+}{\partial \theta_\ell} \left[f_{ij}(\theta_0) -\gamma\right]_+ \\ &= 2 \sum_{j \in \cI_i^0} (f_{ij}(\theta_0) -\gamma) \frac{\partial_-}{\partial \theta_\ell} \left[f_{ij}(\theta_0) -\gamma\right]_+ = 0, \quad \ell=1,\,\ldots,\,n_\theta.
 \end{align*}
 Moreover, in a neighborhood of $\theta_0$ we are either in the situation of Case 1, Case 2, or constant singular values and thus, $\lim_{\theta \to \theta_0} \frac{\partial}{\partial \theta_{\ell}}g_i^0(\theta) = 0$. Thus all partial derivatives are continuous and $g_i^0(\cdot)$ is differentiable.

Finally, since $g(\theta) = \frac{1}{\gamma} \sum_{s_i \in S}\big(g_i^+(\theta) + g_i^0(\theta) + g_i^-(\theta)\big)$, we obtain \eqref{eq:diffg}.
\end{proof}

Property i) establishes a connection between $\loss$ and the $\hinf$ optimization problem, while ii) facilitates gradient-based optimization.

We now outline an algorithm that uses $\loss$ to compute ROMs and reason why a respective optimization of $\loss$ yields ROMs which tend to have a small $\hinf$ error.

Our method consists of 3 main steps: \emph{initialization, $\gamma$-selection}, and \emph{optimization}. During initialization, an initial ROM of a user-selected structure is computed. The initial ROM may be obtained using the FOM and interpolatory MOR or may also be just a random model that meets the required structural constraints. In Section~\ref{sec:implVar}, we provide a comparison of different initialization strategies. After initialization, $\loss$ is minimized for decreasing values of $\gamma$. For that, in the $\gamma$-selection step, $\gamma$ is set to some strictly positive value and the optimization is started. In our \texttt{julia} implementation, in the optimization step, we use the nonlinear optimization engine \textsc{Optim} \cite{mogensen2018optim} configured with the Broyden-Fletcher-Goldfarb-Shanno (BFGS) method for approximating the inverse Hessian matrix and using the line search as implemented in \cite{Hager2006}.

After the optimization of $\loss$ for a fixed $\gamma$, in the $\gamma$-selection step, the optimization can either be called again with an increased or decreased value of $\gamma$ or the process can be stopped. As an example for a $\gamma$-selection step, in our numerical experiments, we provide a fixed sequence of strictly positive and strictly monotonically decreasing values for $\gamma$. We discuss different $\gamma$-selection strategies in Section~\ref{sec:implVar}.

At first, it may seem counter-intuitive to minimize the sum of squares of the error transfer function evaluated at a fixed set of points to minimize its $\hinf$ norm. This is because the $\hinf$ norm is usually attained at a single point or a few points at most. Furthermore, minimizing the sum of squares of the errors usually does not minimize the maximal error value. For this reason, $\loss$ contains the parameter $\gamma$ which functions as a cut-off and allows the modification of the least-squares optimization to get a good $\hinf$ approximation.

In particular, by introducing $\gamma$ we do not directly minimize the $\hinf$ error but instead construct an objective functional that attains its global minimum at zero, when the error at all considered points is below the threshold $\gamma$. Furthermore, the level $\gamma$ allows the distinction between samples that are still relevant for the optimization of $\loss$ and samples that can be ignored at the current iteration. However, note that a zero value of the objective functional does not imply that the $\hinf$ norm of the error is below $\gamma$. On the other hand, if $\loss$ attains a value larger than zero, the $\hinf$ norm is definitely larger than $\gamma$. Therefore, in our numerical experiments, our method is terminated for the first value of $\gamma$ at which after optimization the value of $\loss$ is greater than zero.

\begin{remark}
    In an earlier version of the proposed method, we minimize the simpler objective function
    \begin{align}
      \widetilde{\loss}(\gamma, G, G_r( \cdot, \theta), S) := \frac{1}{\gamma}\sum\limits_{s_i \in S} \left(\left[\left\|G(s_i)-G_r(s_i, \theta)\right\|_2 -\gamma\right]_+ \right)^{2},
    \end{align}
    which only considers the maximal singular value at a given sample point. Note that $\widetilde{\loss}$ depends smoothly on $\theta$ only if the maximal singular value of $G(s_i) - G_r(s_i, \theta)$ is simple for all $s_i$. Using the sum of \emph{all} singular values larger than $\gamma$ instead, we can relax this simplicity assumption.
\end{remark}

\begin{remark}
  This approach has three main benefits compared to directly optimizing the $\hinf$ norm.
  \begin{enumerate}
    \item The $\hinf$ norm depends nonsmoothly on the parameter vector, while the objective functional $\loss$ is differentiable. Therefore, a better convergence behavior can be expected when optimizing $\loss$.
    \item Computing $\hinf$ norms of large-scale transfer functions to a sufficient accuracy is computationally expensive, and due to the nonsmoothness, many evaluations of the $\hinf$ error may be required. On the other hand, \eqref{eq:loss} only requires evaluations of the reduced-order transfer function, since the evaluations of the large-scale transfer function at the sample points can be reused in subsequent iterations.
    \item When using $\loss$ instead of the $\hinf$ norm, we use more information of the transfer function when computing the gradient. Instead of focusing on minimizing the error of the transfer function where it is currently maximized, we consider a set of points where the error is above a certain threshold.
  \end{enumerate}
  Note that when performing a direct $\hinf$ minimization using nonsmooth optimization such as \texttt{GRANSO}, the drawbacks in Point 3 can be somewhat mitigated by an appropriate optimization algorithm. As an example, \texttt{GRANSO} aims at constructing an internal model of the nonsmoothness and by that collects information on the transfer function at the different peaks. 
\end{remark}



\section{Implementation Variants}
\label{sec:implVar}

Each of the three steps in our method, \emph{initialization, $\gamma$-selection}, and \emph{optimization}, can be changed independently. In this sense, we actually propose a template method for obtaining ROMs based on optimization of the ROM parameters. In the following, we briefly outline a few different options for the initialization and the $\gamma$-selection step. For the optimization we have only used BFGS as implemented in \cite{mogensen2018optim}.

\subsection{Initialization}

The task of the initialization step is to provide an initial ROM. The only constraint is that this ROM must satisfy the structural constraints imposed on the final ROM by the user. If the FOM already has the desired structure, then any appropriate structure-preserving MOR method (such as the ones described in Subsection~\ref{subsec:struct}) can be used to obtain an initial ROM. If the FOM does not satisfy the structural constraints or only frequency data is available, then an alternative way that does not use the FOM's system matrices must be used. We have tested three different ways of initialization. The first one is based on a fast greedy structure-preserving MOR method and the following two are a random initialization and a ROM that minimizes the least-squares error over all sample points.

The fast greedy structure-preserving MOR is presented in \cite{beddig2019model}. The main idea is to combine the method for $\hinf$ norm computation from \cite{AliBMSV17a, SchV18} with interpolatory MOR, i.\,e.,  iteratively compute the $\hinf$ norm of the difference between the FOM and the ROM and update the ROM such that it interpolates the FOM at the point on the imaginary axis where the $\hinf$ norm is attained.  While this method is fast and widely applicable, it suffers from a high error. This is due to the fact that the ROM obtained from greedy interpolation is only highly accurate at the interpolation points but does not uniformly approximate the FOM. In our numerical experiments shown in Section~\ref{sec:NumExp}, we have used this as our initialization technique.

The other two initialization methods are based on the generation of random matrices that satisfy the structural constraints. The construction of such matrices (symmetric, skew-symmetric, positive semi-definite) from general random matrices is straightforward. We initialize our method both directly with random systems but also introduce a preliminary least-squares approximation step, in which the function
\begin{align*}
  \LSQ(G,G_r( \cdot, \theta),S):=\sum\limits_{s_i\in S}\left\|G(s_i)-G_r(s_i,\theta)\right\|_2^2,
\end{align*}
where again $G$ and $G_r( \cdot, \theta)$ are transfer functions of the FOM and ROM, respectively, is minimized over $\theta$ using BFGS. The set $S$ contains the same samples that are used during the minimization of $\loss$ after the initialization. In our numerical experiments, these other two initialization techniques have also been tested but they have yielded inferior results, so we do not discuss these further.
 
\subsection{$\gamma$-Selection}

Since our method consists of multiple minimizations of $\loss$ for different values of $\gamma$ to find a ROM with small $\hinf$ error, a method for choosing $\gamma$ must be provided. Furthermore, it is the duty of the $\gamma$-selection step to determine when our method should be terminated. We present two different algorithms for deciding on termination or defining the next $\gamma$ defining $\loss$ in the course of our method. We explicitly assume that the ROM transfer functions of all intermediate iterates are well-defined and elements of $\mathcal{RH}_\infty^{\enu \times \enu}$. Note that this has always been satisfied in our numerical experiments. The problematic cases are those where some of the matrices of our parametrization are only positive semi-definite. For example, the transfer function of the reduced SSO system $\mechSysr$ in \eqref{eq:mechrom} is not well-defined, if, e.\,g., all the matrices $M_r$, $D_r$, and $K_r$ are singular and have a common nonzero vector in their kernels. However, this problem can be solved by adding small multiples of the identity matrix to the respective parametrized ROM matrices such that they are positive definite.

The first $\gamma$-selection method, called \emph{fixed $\gamma$-sequence}, is summarized in Algorithm~\ref{alg:fixed_sequence}. The fixed $\gamma$-sequence approach requires a predefined strictly monotonically decreasing sequence $(\gamma_j)_{j \in \mathbb{N}} \subset \R^+$ to define the optimization problem to be solved in each iteration. If the minimum of $\loss$ for a given $\gamma_j$ is zero, then for all $s_i \in S$, we have that ${\|G(s_i)-G_r(s_i,\theta_j)\|}_2 \le \gamma_j$. Therefore, we terminate the inner optimization loop (in line 2 of Alg.~\ref{alg:fixed_sequence}) when the value of $\loss$ is less than a prescribed \emph{termination tolerance} $\varepsilon$. In this case, the optimization is started again for the smaller value $\gamma_{j+1}$. On the other hand, if the optimization does not succeed in reducing $\loss$ below $\epsilon$, which may happen if it converges to a local (or global) optimum with value greater than $\varepsilon$, or if the number of permitted optimization steps is exceeded, the algorithm is terminated. 
In the numerical experiments reported in Section~\ref{sec:NumExp}, we have used Alg.~\ref{alg:fixed_sequence} with termination tolerance $\varepsilon = 10^{-14}$. The scaling factor $1/\gamma$ in \eqref{eq:loss} is introduced to make the final loss values, that are obtained after each optimization for the different values of $\gamma$, comparable such that the same termination tolerance can be used for $\gamma$ varying over several orders of magnitude.

\begin{algorithm}[tb]
  \LinesNumbered
  \SetAlgoLined
  \DontPrintSemicolon
  \SetKwInOut{Input}{Input}\SetKwInOut{Output}{Output}
  \Input{FOM transfer function $G\in\mathcal{RH}_\infty^{n_u \times n_u}$, initial parametrized ROM transfer function $G_r( \cdot, \theta_0) \in \rhinf^{n_u \times n_u}$ with parameter $\theta_0 \in \R^{n_\theta}$, sample point set $S \subset \mathrm{i}\R$, strictly monotonically decreasing sequence $(\gamma_j)_{j\in \mathbb{N}} \subset \R^+$, termination tolerance $\varepsilon > 0$}
  \Output{final ROM parameters $\theta_{\rm fin} \in \R^{n_\theta}$}
  Set $j:=0$ and $\alpha_0:=0$.\;
  \While{$\alpha_{j} \le \varepsilon$}{
    Solve the minimization problem $\alpha_{j+1} := \min_{\theta\in \R^{n_\theta}}\loss(\gamma_{j+1},G,G_r( \cdot, \theta),S)$ with minimizer $\theta_{j+1} \in \R^{n_\theta}$, initialized at $\theta_{j}$. \;
    Set $j := j+1$. \;
  }
  Set $\theta_{\rm fin}:=\theta_{j-1}$.\;
  \caption{Fixed $\gamma$-sequence}
  \label{alg:fixed_sequence}
\end{algorithm}

Another approach is to select the sequence $(\gamma)_{j \in \N}$ by bisection. This is outlined in Algorithm~\ref{alg:bisection}. The difference from the previous $\gamma$ selection approach is that here, the user only needs to provide an upper bound $\gamma_{\rm u}$ and then $\gamma$ is adjusted based on the results of the optimization of the ROM parameters. As before, the termination tolerance $\varepsilon_2$ is the maximal final objective value at which $\gamma$ is still reduced. If the relative distance between $\gamma_{\rm u}$ and $\gamma_{\rm l}$ is below the bisection tolerance $\varepsilon_1$, then the bisection is stopped. We choose the relative difference, since the final accuracy that can be achieved is not known beforehand. Choosing the absolute difference between $\gamma_{\rm u}$ and $\gamma_{\rm l}$ as stopping condition can thus lead to a premature termination or unnecessarily many iterations. In our experiments this algorithm takes longer than Algorithm~\ref{alg:fixed_sequence} and leads to slightly larger $\hinf$ errors. However bisection is succesfully used in \cite{SchV2021Adaptive} in conjunction with an adaptive sample point selection.

\begin{algorithm}[tbh]
  \LinesNumbered
  \SetAlgoLined
  \DontPrintSemicolon
  \SetKwInOut{Input}{Input}\SetKwInOut{Output}{Output}
  \Input{FOM transfer function $G\in\mathcal{RH}_{\infty}^{n_u \times n_u}$, initial ROM transfer function $G_r( \cdot, \theta_0) \in \rhinf^{n_u \times n_u}$ with parameter $\theta_0 \in \R^{n_\theta}$, sample point set $S \subset \mathrm{i}\R$, upper bound $\gamma_{\rm u} > 0$, bisection tolerance $\varepsilon_1 > 0$, termination tolerance $\varepsilon_2 > 0$}
  \Output{final ROM parameters $\theta_{\rm fin} \in \R^{n_\theta}$}
  Set $j:=0$ and $\gamma_{\rm l}:=0$.\;
  \While{$(\gamma_{\rm u}-\gamma_{\rm l})/(\gamma_{\rm u}+\gamma_{\rm l}) > \varepsilon_1$}{
    Set $\gamma=(\gamma_{\rm u}+\gamma_{\rm l})/2$.\;
    Solve the minimization problem $\alpha_{j+1} := \min_{\theta\in \R^{n_\theta}}\loss(\gamma,G,G_r( \cdot, \theta),S)$ with minimizer $\theta_{j+1} \in \R^{n_\theta}$, initialized at $\theta_j$. \;
  \eIf{$\alpha_{j+1} > \varepsilon_2$}{
    Set $\gamma_{\rm l}:=\gamma$.\;
    }{
    Set $\gamma_{\rm u}:=\gamma$.\;
  }
  Set $j:=j+1$.
  }
  Set $\theta_{\rm fin} := \theta_j$.\;
  \caption{$\gamma$-Bisection}
  \label{alg:bisection}
\end{algorithm}

\section{Numerical Experiments}
\label{sec:NumExp}

In the following, we compare the approximation quality, in terms of the $\hinf$ norm and the $\htwo$ norm, of our method (SOBMOR) to other, previously established, methods that also preserve the structure of either a pH or SSO system. For both model types we provide a comparison of our method with the structure-preserving MOR methods presented in Section~\ref{sec:background}.
For both model structures optimization-based MOR outperforms the other structure-preserving methods by a few orders of magnitude with respect to accuracy in terms of the $\hinf$ norm.

\subsection{Port-Hamiltonian Systems}

The pH model we consider is used in~\cite{Gugercin2012} to showcase the effectiveness of pH-IRKA. It models a mass-spring-damper system. The inputs are two forces applied to the first two masses and the outputs are their respective velocities. For details on the setup of the system matrices we refer to~\cite{Gugercin2012}. In our experiments, we follow \cite{Gugercin2012} and first use a FOM with the moderate state dimension of $100$ in order to deliver a thorough analysis of the error based on $\hinf$ and $\htwo$ norms. Then a larger state-space dimension of $20,000$ is used to show the effectiveness of SOBMOR also in the large-scale case.

For the sample points, we use 800 logarithmically spaced points between $10^{-4}$ and $10^{3}$. We add a few sample points outside of this interval, namely $0, 10^{-8}, 10^{-7}, 10^{-6}, 10^4, 10^5$, and $10^6$. The results shown here are obtained with Algorithm~\ref{alg:fixed_sequence}, in which the objective function $L$ is minimized for a sequence of monotonically decreasing ${(\gamma_j)}_{j\in\mathbb{N}}$ until the computed minimum is different from zero (up to the termination tolerance). Here we choose ${(\gamma_j)}_{j\in\mathbb{N}}$ as a sequence of 300 logarithmically spaced points between $10^{-1}$ and $10^{-14}$.

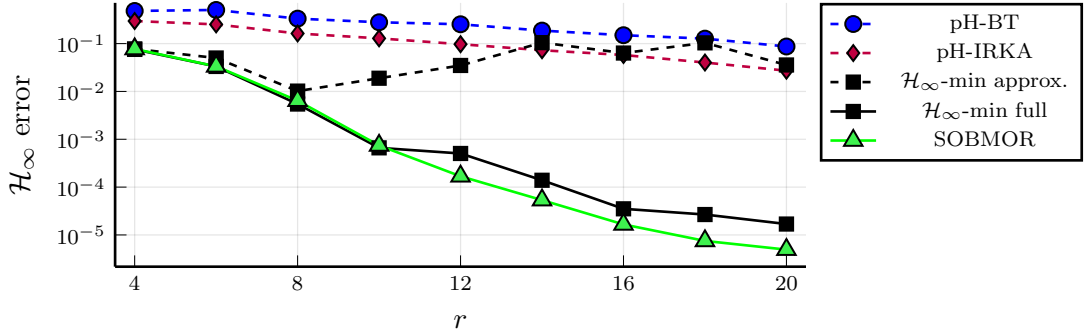
\begin{figure}[tb]
  \centering
  \begin{tikzpicture}[/tikz/background rectangle/.style={fill={rgb,1:red,1.0;green,1.0;blue,1.0}, draw opacity={1.0}}, show background rectangle]
\begin{axis}[title={}, title style={at={{(0.5,1)}}, font={{\fontsize{14 pt}{18.2 pt}\selectfont}}, color={rgb,1:red,0.0;green,0.0;blue,0.0}, draw opacity={1.0}, rotate={0.0}}, legend style={color={rgb,1:red,0.0;green,0.0;blue,0.0}, draw opacity={1.0}, line width={1}, solid, fill={rgb,1:red,1.0;green,1.0;blue,1.0}, fill opacity={1.0}, text opacity={1.0}, font={{\fontsize{8 pt}{10.4 pt}\selectfont}}, at={(1.02, 1)}, anchor={north west}}, axis background/.style={fill={rgb,1:red,1.0;green,1.0;blue,1.0}, opacity={1.0}}, anchor={north west}, xshift={1.0mm}, yshift={-1.0mm}, width={0.7\textwidth}, height={0.25\textheight}, scaled x ticks={false}, xlabel={$r$}, x tick style={color={rgb,1:red,0.0;green,0.0;blue,0.0}, opacity={1.0}}, x tick label style={color={rgb,1:red,0.0;green,0.0;blue,0.0}, opacity={1.0}, rotate={0}}, xlabel style={, font={{\fontsize{11 pt}{14.3 pt}\selectfont}}, color={rgb,1:red,0.0;green,0.0;blue,0.0}, draw opacity={1.0}, rotate={0.0}}, xmajorgrids={true}, xmin={3.52}, xmax={20.48}, xtick={{4.0,8.0,12.0,16.0,20.0}}, xticklabels={{$4$,$8$,$12$,$16$,$20$}}, xtick align={inside}, xticklabel style={font={{\fontsize{8 pt}{10.4 pt}\selectfont}}, color={rgb,1:red,0.0;green,0.0;blue,0.0}, draw opacity={1.0}, rotate={0.0}}, x grid style={color={rgb,1:red,0.0;green,0.0;blue,0.0}, draw opacity={0.1}, line width={0.5}, solid}, axis x line*={left}, x axis line style={color={rgb,1:red,0.0;green,0.0;blue,0.0}, draw opacity={1.0}, line width={1}, solid}, scaled y ticks={false}, ylabel={$\hinf$ error}, y tick style={color={rgb,1:red,0.0;green,0.0;blue,0.0}, opacity={1.0}}, y tick label style={color={rgb,1:red,0.0;green,0.0;blue,0.0}, opacity={1.0}, rotate={0}}, ylabel style={, font={{\fontsize{11 pt}{14.3 pt}\selectfont}}, color={rgb,1:red,0.0;green,0.0;blue,0.0}, draw opacity={1.0}, rotate={0.0}}, ymode={log}, log basis y={10}, ymajorgrids={true}, ymin={2.169570177021927e-6}, ymax={0.7106243024441214}, ytick={{1.0e-5,0.0001,0.001,0.01,0.1}}, yticklabels={{$10^{-5}$,$10^{-4}$,$10^{-3}$,$10^{-2}$,$10^{-1}$}}, ytick align={inside}, yticklabel style={font={{\fontsize{8 pt}{10.4 pt}\selectfont}}, color={rgb,1:red,0.0;green,0.0;blue,0.0}, draw opacity={1.0}, rotate={0.0}}, y grid style={color={rgb,1:red,0.0;green,0.0;blue,0.0}, draw opacity={0.1}, line width={0.5}, solid}, axis y line*={left}, y axis line style={color={rgb,1:red,0.0;green,0.0;blue,0.0}, draw opacity={1.0}, line width={1}, solid}, colorbar style={title={}}, point meta max={nan}, point meta min={nan}]
    \addplot[color=blue, name path={fdc86282-b401-4979-abae-4468f1aaf760}, draw opacity={1.0}, line width={1}, solid, mark={*}, mark size={3.0 pt}, mark options={color={rgb,1:red,0.0;green,0.0;blue,0.0}, draw opacity={1.0}, fill=blue, fill opacity={1.0}, line width={0.75}, rotate={0}, solid}, dashed]
        coordinates {
            (4,0.4813471942945101)
            (6,0.5053326346940811)
            (8,0.33265402950983736)
            (10,0.2792712148058059)
            (12,0.2535149168194241)
            (14,0.1870396561695439)
            (16,0.14977502520076516)
            (18,0.12700565585039805)
            (20,0.086946101277349)
        }
        ;
    \addlegendentry {pH-BT}
    \addplot[color=purple, name path={5e7e1fa2-9a18-4027-b9a8-5d99613a1c14}, draw opacity={1.0}, line width={1}, solid, mark={diamond*}, mark size={3.0 pt}, mark options={color=black, draw opacity={1.0}, fill=purple, fill opacity={1.0}, line width={0.75}, rotate={0}, solid}, dashed]
        coordinates {
            (4,0.29644543479395774)
            (6,0.2520497943209842)
            (8,0.16276581594136239)
            (10,0.1287240783673755)
            (12,0.09709961195728065)
            (14,0.07263289970224365)
            (16,0.05774756684325709)
            (18,0.04021116428444131)
            (20,0.027135580825011017)
        }
        ;
    \addlegendentry {pH-IRKA}
    \addplot[color=black, name path={0e11bac5-ad65-41ff-b905-53a000898e6c}, draw opacity={1.0}, line width={1}, solid, mark={square*}, mark size={2.5 pt}, mark options={color=black, draw opacity={1.0}, fill=black, fill opacity={1.0}, line width={0.75}, rotate={0}, solid}, dashed]
        coordinates {
            (4,0.07735495023367092)
            (6,0.050013119325081054)
            (8,0.010166045719295209)
            (10,0.018887457655148494)
            (12,0.034907418680167855)
            (14,0.10475998909237694)
            (16,0.06348716318338848)
            (18,0.1036321884150996)
            (20,0.03574214744151978)
        }
        ;
      \addlegendentry {$\hinf$-min approx.}
    \addplot[color=black, name path={0e11bac5-ad65-41ff-b905-53a000898e6c}, draw opacity={1.0}, line width={1}, solid, mark={square*}, mark size={2.5 pt}, mark options={color=black, draw opacity={1.0}, fill=black, fill opacity={1.0}, line width={0.75}, rotate={0}, solid}]
        coordinates {
          (4, 7.568281e-02)
          (6, 3.329279e-02)
          (8, 5.452303e-03)
          (10,6.617802e-04)
          (12,5.047309e-04)
          (14,1.390656e-04)
          (16,3.517546e-05)
          (18,2.670271e-05)
          (20,1.690427e-05)
        }
        ;
      \addlegendentry {$\hinf$-min full}
    \addplot[color=green, name path={ee35c385-852e-4ea4-bc93-4a19681d6410}, draw opacity={1.0}, line width={1}, solid, mark={triangle*}, mark size={4.0 pt}, mark options={color=black, draw opacity={1.0}, fill=ourcolor, fill opacity={1.0}, line width={0.75}, rotate={0}, solid}]
        coordinates {
          (4, 7.716396e-02)
          (6, 3.342240e-02)
          (08,6.411867e-03)
          (10,7.429191e-04)
          (12,1.685170e-04)
          (14,5.347107e-05)
          (16,1.659474e-05)
          (18,7.503271e-06)
          (20,4.931897e-06)
        }
        ;
    \addlegendentry {SOBMOR}
\end{axis}
\end{tikzpicture}
  \caption{$\hinf$ errors for different MOR methods for varying ROM orders. The FOM is the mass-spring-damper pH system in \cite{Gugercin2012} of order 100.}%
  \label{fig:ph_hinferrors}
\end{figure}

We compare our method to the previous pH-MOR algorithms pH-BT and pH-IRKA as well as to the two variants of direct $\hinf$ error minimization using \texttt{GRANSO} (configured with its standard optimization parameters) in conjunction with either \texttt{linorm\_subsp} or \texttt{ab13hd} to compute the $\hinf$ errors to be minimized. We emphasize that \texttt{ab13hd} is limited to FOMs with a moderate state-space dimension and cannot exploit sparsity. In Fig.~\ref{fig:ph_hinferrors} the $\hinf$ errors are shown for different reduced model orders. The parameter optimization methods yield a better accuracy compared to the previous methods for small reduced model orders (4--8). However, the direct $\hinf$ error minimization approach that is based on \texttt{linorm\_subsp} yields worse results as the reduced model order increases beyond $r=8$. This is due to inaccuracies in the $\hinf$ norm approximation by \texttt{linorm\_subsp}, which are known to occur when the transfer function has many different peaks with similar magnitude. On the other hand, when using our method and the direct $\hinf$ minimization based on \texttt{ab13hd} (which is computationally prohibitive for larger FOMs, i.\,e., the most common use-case of MOR) the error continues to decrease for larger reduced orders. Finally, these methods outperform all other structure-preserving methods by three orders of magnitude.
For the reduced model orders between 12 and 20 our method outperforms both direct $\hinf$ minimization approaches.

\begin{figure}[bt]
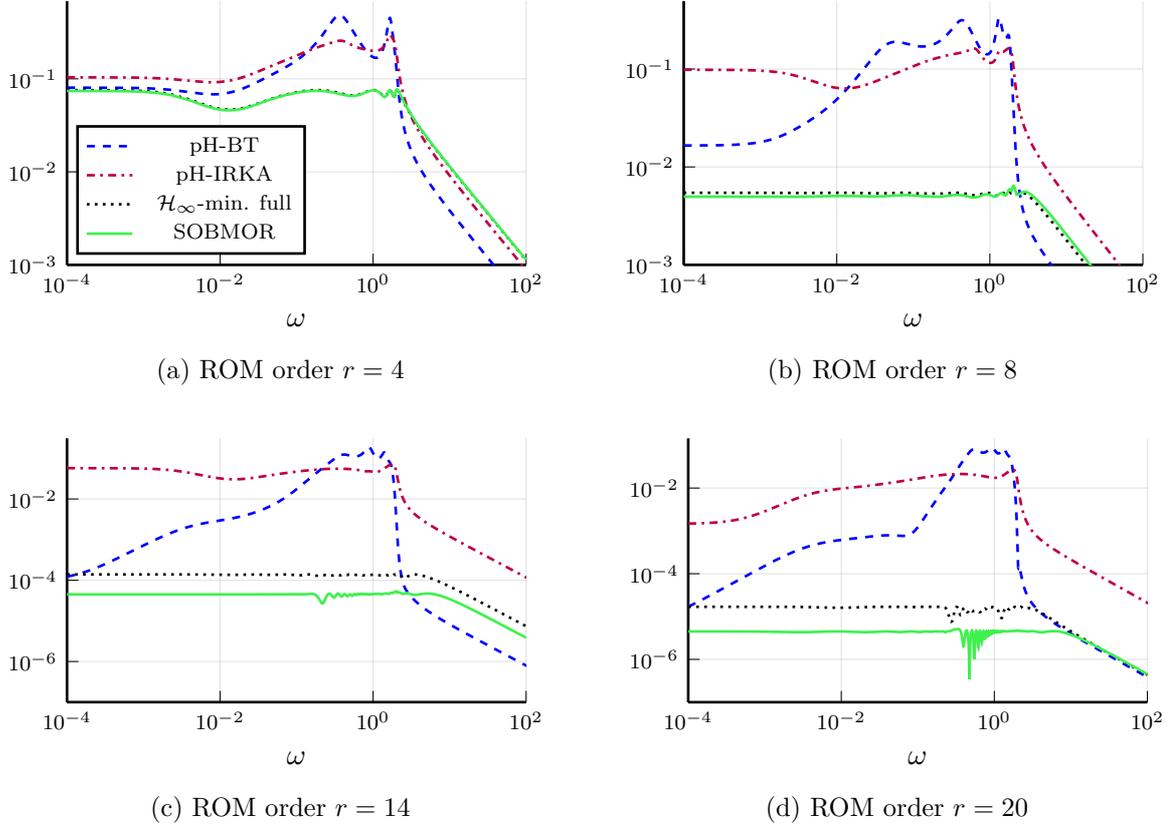

  \centering
  \begin{tabular}{cc}
    \input{./PlotSources/ErrorTransfunPlots/PH/MethodComparison/errors4.tex} &
    \input{./PlotSources/ErrorTransfunPlots/PH/MethodComparison/errors8.tex} \\
    (a) ROM order $r=4$ & (b) ROM order $r=8$ \vspace{0.5cm}\\ 
    \input{./PlotSources/ErrorTransfunPlots/PH/MethodComparison/errors14.tex} &
    \input{./PlotSources/ErrorTransfunPlots/PH/MethodComparison/errors20.tex}\\
    (c) ROM order $r=14$ & (d) ROM order $r=20$ \\
  \end{tabular}
  \caption{Maximal singular value of the error transfer function for different 
  MOR methods as $r$ is varied. The FOM is the mass-spring-damper pH system in \cite{Gugercin2012} of order 100.}%
  \label{fig:ph_errTF_methods}
\end{figure}

In Figure~\ref{fig:ph_errTF_methods}, the maximum singular values of the transfer function of the error between the given model and its approximations obtained from the different MOR methods are shown for different $r$. As $r$ increases, the error obtained with our method tends to not vary much over different orders of magnitudes and has a well-balanced accuracy over the imaginary axis. This desired behavior of the error is not achieved with the other methods.  

Furthermore, it can be observed that the approximation is consistently better for most frequencies. This implies that the larger $\hinf$ error which the other methods produce is not just due to a few areas on the imaginary axis where the transfer function is approximated poorly. Therefore, even though our method is designed to minimize the $\hinf$ error, in the comparison of the $\htwo$ errors in Figure~\ref{fig:ph_h2errors}, the approximation error is again several orders of magnitude smaller when using our method.

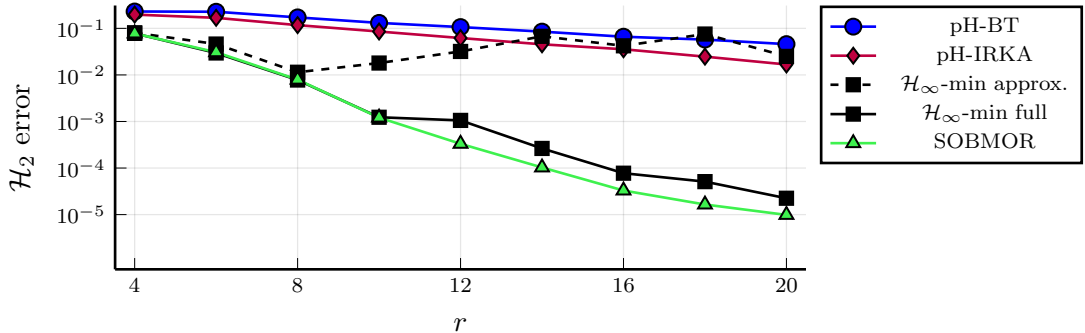
\begin{figure}[tbh]
  \centering
  \begin{tikzpicture}[/tikz/background rectangle/.style={fill={rgb,1:red,1.0;green,1.0;blue,1.0}, draw opacity={1.0}}, show background rectangle]
\begin{axis}[title={}, title style={at={{(0.5,1)}}, font={{\fontsize{14 pt}{18.2 pt}\selectfont}}, color={rgb,1:red,0.0;green,0.0;blue,0.0}, draw opacity={1.0}, rotate={0.0}}, legend style={color={rgb,1:red,0.0;green,0.0;blue,0.0}, draw opacity={1.0}, line width={1}, solid, fill={rgb,1:red,1.0;green,1.0;blue,1.0}, fill opacity={1.0}, text opacity={1.0}, font={{\fontsize{8 pt}{10.4 pt}\selectfont}}, at={(1.02, 1)}, anchor={north west}}, axis background/.style={fill={rgb,1:red,1.0;green,1.0;blue,1.0}, opacity={1.0}}, anchor={north west}, xshift={1.0mm}, yshift={-1.0mm}, width={0.7\textwidth}, height={0.25\textheight}, scaled x ticks={false}, xlabel={$r$}, x tick style={color={rgb,1:red,0.0;green,0.0;blue,0.0}, opacity={1.0}}, x tick label style={color={rgb,1:red,0.0;green,0.0;blue,0.0}, opacity={1.0}, rotate={0}}, xlabel style={, font={{\fontsize{11 pt}{14.3 pt}\selectfont}}, color={rgb,1:red,0.0;green,0.0;blue,0.0}, draw opacity={1.0}, rotate={0.0}}, xmajorgrids={true}, xmin={3.52}, xmax={20.48}, xtick={{4.0,8.0,12.0,16.0,20.0}}, xticklabels={{$4$,$8$,$12$,$16$,$20$}}, xtick align={inside}, xticklabel style={font={{\fontsize{8 pt}{10.4 pt}\selectfont}}, color={rgb,1:red,0.0;green,0.0;blue,0.0}, draw opacity={1.0}, rotate={0.0}}, x grid style={color={rgb,1:red,0.0;green,0.0;blue,0.0}, draw opacity={0.1}, line width={0.5}, solid}, axis x line*={left}, x axis line style={color={rgb,1:red,0.0;green,0.0;blue,0.0}, draw opacity={1.0}, line width={1}, solid}, scaled y ticks={false}, ylabel={$\htwo$ error}, y tick style={color={rgb,1:red,0.0;green,0.0;blue,0.0}, opacity={1.0}}, y tick label style={color={rgb,1:red,0.0;green,0.0;blue,0.0}, opacity={1.0}, rotate={0}}, ylabel style={, font={{\fontsize{11 pt}{14.3 pt}\selectfont}}, color={rgb,1:red,0.0;green,0.0;blue,0.0}, draw opacity={1.0}, rotate={0.0}}, ymode={log}, log basis y={10}, ymajorgrids={true}, ymin={6.67842871883e-7}, ymax={0.3096736571595799}, ytick={{1.0e-5,0.0001,0.001,0.01,0.1}}, yticklabels={{$10^{-5}$,$10^{-4}$,$10^{-3}$,$10^{-2}$,$10^{-1}$}}, ytick align={inside}, yticklabel style={font={{\fontsize{8 pt}{10.4 pt}\selectfont}}, color={rgb,1:red,0.0;green,0.0;blue,0.0}, draw opacity={1.0}, rotate={0.0}}, y grid style={color={rgb,1:red,0.0;green,0.0;blue,0.0}, draw opacity={0.1}, line width={0.5}, solid}, axis y line*={left}, y axis line style={color={rgb,1:red,0.0;green,0.0;blue,0.0}, draw opacity={1.0}, line width={1}, solid}, colorbar style={title={}}, point meta max={nan}, point meta min={nan}]
    \addplot[color=blue, name path={c29ff3de-8649-46b3-b462-89001bcb934e}, draw opacity={1.0}, line width={1}, solid, mark={*}, mark size={3.0 pt}, mark options={color={rgb,1:red,0.0;green,0.0;blue,0.0}, draw opacity={1.0}, fill=colorone, fill opacity={1.0}, line width={0.75}, rotate={0}, solid}]
        coordinates {
            (4,0.23017649150521305)
            (6,0.22748024019224755)
            (8,0.17200427306859975)
            (10,0.13106252746419725)
            (12,0.10657813594918993)
            (14,0.0851264732265212)
            (16,0.06614633866271054)
            (18,0.05744663877285041)
            (20,0.045876845964380734)
        }
        ;
    \addlegendentry {pH-BT}
    \addplot[color=purple, name path={2881e45f-97b1-4b88-91a0-f19f351aad0d}, draw opacity={1.0}, line width={1}, solid, mark={diamond*}, mark size={3.0 pt}, mark options={color={rgb,1:red,0.0;green,0.0;blue,0.0}, draw opacity={1.0}, fill=purple, fill opacity={1.0}, line width={0.75}, rotate={0}, solid}]
        coordinates {
            (4,0.19764106247068275)
            (6,0.16847013612811368)
            (8,0.11734529307456867)
            (10,0.08550086038303994)
            (12,0.06200037028121121)
            (14,0.04545206701894334)
            (16,0.035583538449677314)
            (18,0.024692098762706182)
            (20,0.016756190657142143)
        }
        ;
    \addlegendentry {pH-IRKA}
    \addplot[color=black, name path={f86e0b97-dce2-4be9-a13e-057d46130b97}, draw opacity={1.0}, line width={1}, solid, mark={square*}, mark size={2.5 pt}, mark options={color={rgb,1:red,0.0;green,0.0;blue,0.0}, draw opacity={1.0}, fill=black, fill opacity={1.0}, line width={0.75}, rotate={0}, solid}, dashed]
        coordinates {
            (4,0.08040156264594706)
            (6,0.04606247171068631)
            (8,0.011419987994304938)
            (10,0.018004720128258943)
            (12,0.03193855638526959)
            (14,0.06751223135142101)
            (16,0.04218987264154491)
            (18,0.07556790995374629)
            (20,0.025001114641358596)
        }
        ;
    \addlegendentry {$\hinf$-min approx.}
    \addplot[color=black, name path={f86e0b97-dce2-4be9-a13e-057d46130b97}, draw opacity={1.0}, line width={1}, solid, mark={square*}, mark size={2.5 pt}, mark options={color={rgb,1:red,0.0;green,0.0;blue,0.0}, draw opacity={1.0}, fill=black, fill opacity={1.0}, line width={0.75}, rotate={0}, solid}]
        coordinates {
            (4,0.0780290668546802)
            (6,0.02959373208802068)
            (8,0.007772866619934149)
            (10,0.0012327061520012877)
            (12,0.0010550239907123156)
            (14,0.00026360486662181307)
            (16,7.74934974286619e-5)
            (18,5.1078315489733386e-5)
            (20,2.246048006973132e-5)
        }
        ;
    \addlegendentry {$\hinf$-min full}
    \addplot[color=ourcolor, name path={0cb40dd5-bd1a-472b-a671-b156d11f5b21}, draw opacity={1.0}, line width={1}, solid, mark={triangle*}, mark size={3.0 pt}, mark options={color={rgb,1:red,0.0;green,0.0;blue,0.0}, draw opacity={1.0}, fill=ourcolor, fill opacity={1.0}, line width={0.75}, rotate={0}, solid}]
      coordinates {
        (4,0.07771332110601026)
        (6,0.030733233145957082)
        (8,0.00791980283324156)
        (10,0.0012121989700079382)
        (12,0.00033156096069169227)
        (14,0.00010293402759170686)
        (16,3.29103118554293e-5)
        (18,1.647938188774737e-5)
        (20,9.850772084654591e-6)
      }
        ;
    \addlegendentry {SOBMOR}
\end{axis}
\end{tikzpicture}
  \caption{$\htwo$ error for different MOR methods as ROM order $r$ is varied. The FOM is the mass-spring-damper pH model in \cite{Gugercin2012} of order 100.}%
  \label{fig:ph_h2errors}
\end{figure}

In order to illustrate the performance of SOBMOR also in large-scale settings, for the next experiments we use the FOM dimension $n=20,000$. In Figure~\ref{fig:ph_large_comparison}, we compare the performance of our method against pH-IRKA, since both methods rely only on sparse linear solves during the initial sampling and are thus appropriate for the large-scale case. With SOBMOR, we again obtain an error that is lower by a few orders of magnitude. In fact, in \cite{Gugercin2012}, it is reported that the $\hinf$ error is still above $10^{-4}$ for ROMs of order 50 when using pH-IRKA.



In our experiments, the runtime of our algorithm is mainly driven by the choice of the ROM order. Table~\ref{tab:runtime_comparison} shows a comparison of the runtime for different ROM orders. As expected, the runtime of pH-BT does not vary with $r$ as its main computational burden is only dependent on the FOM dimension. On the other hand the runtime of pH-IRKA is mostly dependent on the number of fixed-point iterations until convergence. In our experiments the number of iterations was not correlated to $r$, which explains the similar runtimes. The optimization-based approaches tend to take more time, as the reduced model order increases. 
However, the runtime does not increase monotonically, as $r$ increases.
Note that our method takes significantly less time compared to the direct $\hinf$ minimization approach for all tested reduced model orders.

\begin{table}[htpb]
  \centering
  \caption{Runtimes (in seconds) for different MOR algorithms on the pH model in \cite{Gugercin2012} of order 100 for different reduced model orders $r$.}
  \label{tab:runtime_comparison}
  \begin{tabular}{c|cccc}
    $r$ & pH-BT & pH-IRKA & $\hinf$-min. full & SOBMOR \\ \hline
    4 &  8.75e$-$03 & 4.56e$-$02 & 1.79e+02 & 1.78e+00 \\
    6 &  8.91e$-$03 & 4.58e$-$02 & 1.39e+03 & 6.99e+01 \\
    8 &  8.90e$-$03 & 4.50e$-$02 & 7.10e+02 & 2.27e+01 \\
    10 & 8.83e$-$03 & 4.53e$-$02 & 7.84e+02 & 7.45e+01 \\
    12 & 8.84e$-$03 & 4.57e$-$02 & 2.92e+02 & 1.86e+02 \\
    14 & 8.85e$-$03 & 4.66e$-$02 & 3.32e+02 & 2.83e+02 \\
    16 & 8.84e$-$03 & 4.57e$-$02 & 6.44e+03 & 3.90e+02 \\
    18 & 8.85e$-$03 & 4.55e$-$02 & 6.52e+03 & 1.15e+03 \\
    20 & 8.85e$-$03 & 4.59e$-$02 & 5.18e+03 & 8.08e+02 \\
  \end{tabular}
\end{table}

The strong dependence of our method's runtime on $r$ is mostly due to the more expensive gradient computation and increased number of iterations for the larger ROMs. On the other hand, the FOM only has to be sampled once in the beginning. In fact, even during the experiments in which the FOM dimension is set to $20,000$ the runtime ranges between 19.3s (for $r=4$) and 932.10s (for $r=16$). Thus, the FOM order does not have a major influence on the runtime in our experiments. However, in some settings the sampling procedure may also be computationally taxing, such that the number of sample points must be kept small. For that, we present an adaptive sampling strategy for our method in \cite{SchV2021Adaptive}.

Since pH-IRKA converged in well under a second in each experiment, our method is not yet competitive with respect to runtime. However, since MOR is usually performed offline to speed up simulations or the computation of a control signal, the greatly increased accuracy well compensates for the increased runtime. 
Note that as $\gamma$ is reduced during the course of our method, the number of iterations during the gradient-based optimization increases. This behavior can be observed in Figure~\ref{fig:iters_over_gamma}. This usually makes the final reduction steps the most expensive ones. When a smaller runtime is desired, a desired level of accuracy leading to a final value of $\gamma$ may be prescribed. Then, the final and most expensive optimizations may be avoided. In this way, we can trade off a lower accuracy of the ROM for a much faster termination of our method.

\begin{figure}[bth]
  \centering
  \begin{tabular}{cc}
    \input{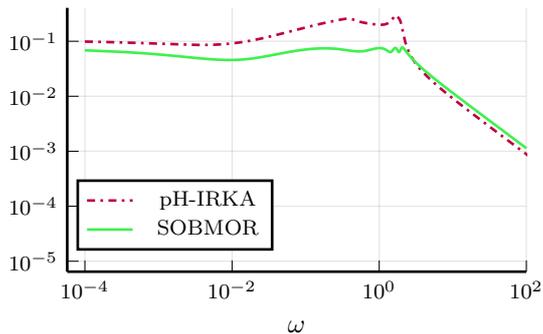} &
    \input{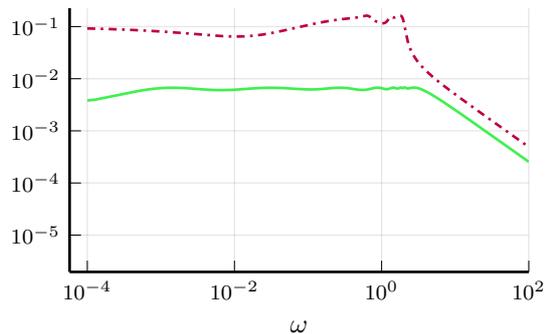} \\
    (a) ROM order $r=4$ & (b) ROM order $r=8$ \vspace{0.5cm}\\ 
    \input{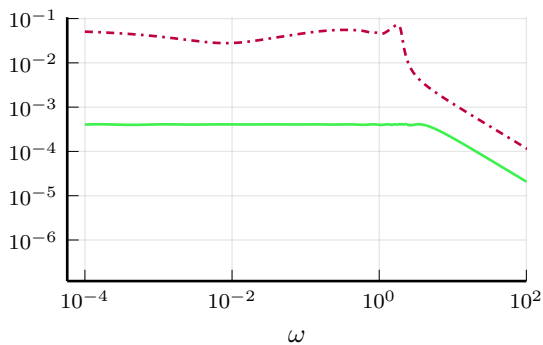} &
    \input{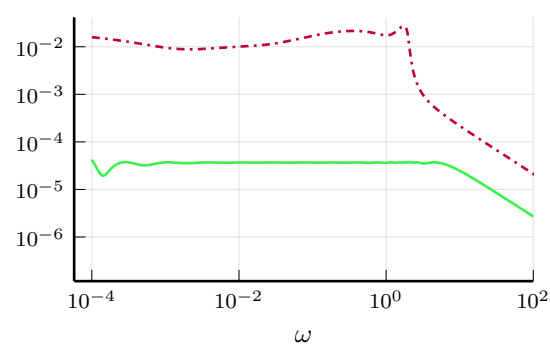}\\
    (c) ROM order $r=14$ & (d) ROM order $r=20$ \\
  \end{tabular}
  \caption{Maximum singular value of the error transfer function for different MOR methods as $r$ is varied. The FOM is the mass-spring-damper pH system in \cite{Gugercin2012} of order $20,000$. Note that pH-BT is computationally prohibitively expensive for such large FOM orders, since the full balancing transformation (which is needed for pH-BT) cannot be computed efficiently. Therefore, we show no results for pH-BT.}%
  \label{fig:ph_large_comparison}
\end{figure}

\begin{figure}[bth]
  \centering
  \input{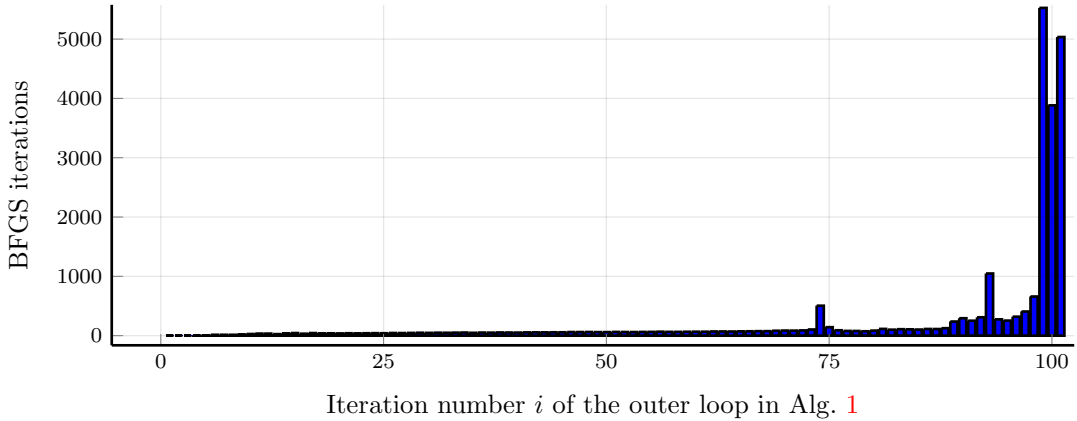}
  \caption{Number of iterations inside the optimization for subsequent accuracy levels $\gamma_i$ for $r=20$.}%
  \label{fig:iters_over_gamma}
\end{figure}

\subsection{Second-Order Systems}

The scalable mechanical system that we use to showcase the effectiveness of our method for second-order systems is introduced in \cite{Truhar2009} as a triple mass-spring-damper chain. We refer to \cite{Truhar2009} for a detailed description of the mass-spring-damper chain model as well as the setup of the system matrices $M$, $D$, and $K$. In order to compare our method with second-order balanced truncation and to compute the $\hinf$ norm using \cite{Boyd1990} for a globally valid result, we choose only a medium state dimension ($n_x=301$) for our experiments. Since in \cite{Truhar2009}, no input or output mapping is provided, we set $B:=\begin{bmatrix}
    I_3 & 0_{n_x-3\times 3}
  \end{bmatrix}^\T$,
which corresponds to actuation and measurement at the first three masses. For the sample points we choose 300 logarithmically spaced points between $10^{-4}$ and $10^{3}$. Again, we add a few points outside of this interval, namely $0, 10^{-8}, 10^{-7}, 10^{-6}, 10^{-4}, 10^4, 10^5$, and $10^6$. The $\gamma$ sequence is chosen as in the previously described experiments. 

Next to the full parametrization for second-order systems introduced in Lemma~\ref{lem:sndparam}, we also conduct experiments, in which $M(\theta)$ is a positive semidefinite diagonal matrix. In this way, we can use fewer parameters at (theoretically) no loss in accuracy, since any symmetric second-order system $\Sigma_{sso}$ can be written with diagonal $M$ by a change of basis, while preserving symmetries of the system.

\begin{figure}[bth]
  \centering
  \begin{tikzpicture}[/tikz/background rectangle/.style={fill={rgb,1:red,1.0;green,1.0;blue,1.0}, draw opacity={1.0}}, show background rectangle]
\begin{axis}[title={}, title style={at={{(0.5,1)}}, font={{\fontsize{14 pt}{18.2 pt}\selectfont}}, color={rgb,1:red,0.0;green,0.0;blue,0.0}, draw opacity={1.0}, rotate={0.0}}, legend style={color={rgb,1:red,0.0;green,0.0;blue,0.0}, draw opacity={1.0}, line width={1}, solid, fill={rgb,1:red,1.0;green,1.0;blue,1.0}, fill opacity={1.0}, text opacity={1.0}, font={{\fontsize{8 pt}{10.4 pt}\selectfont}}, at={(1.02, 1)}, anchor={north west}}, axis background/.style={fill={rgb,1:red,1.0;green,1.0;blue,1.0}, opacity={1.0}}, anchor={north west}, xshift={1.0mm}, yshift={-1.0mm}, width={0.8\textwidth}, height={0.25\textheight}, scaled x ticks={false}, xlabel={$r$}, x tick style={color={rgb,1:red,0.0;green,0.0;blue,0.0}, opacity={1.0}}, x tick label style={color={rgb,1:red,0.0;green,0.0;blue,0.0}, opacity={1.0}, rotate={0}}, xlabel style={, font={{\fontsize{11 pt}{14.3 pt}\selectfont}}, color={rgb,1:red,0.0;green,0.0;blue,0.0}, draw opacity={1.0}, rotate={0.0}}, xmajorgrids={true}, xmin={4.52}, xmax={21.48}, xtick={{5.0,9.0,13.0,17.0,21.0}}, xticklabels={{$5$,$9$,$13$,$17$,$21$}}, xtick align={inside}, xticklabel style={font={{\fontsize{8 pt}{10.4 pt}\selectfont}}, color={rgb,1:red,0.0;green,0.0;blue,0.0}, draw opacity={1.0}, rotate={0.0}}, x grid style={color={rgb,1:red,0.0;green,0.0;blue,0.0}, draw opacity={0.1}, line width={0.5}, solid}, axis x line*={left}, x axis line style={color={rgb,1:red,0.0;green,0.0;blue,0.0}, draw opacity={1.0}, line width={1}, solid}, scaled y ticks={false}, ylabel={$\hinf$ error}, y tick style={color={rgb,1:red,0.0;green,0.0;blue,0.0}, opacity={1.0}}, y tick label style={color={rgb,1:red,0.0;green,0.0;blue,0.0}, opacity={1.0}, rotate={0}}, ylabel style={, font={{\fontsize{11 pt}{14.3 pt}\selectfont}}, color={rgb,1:red,0.0;green,0.0;blue,0.0}, draw opacity={1.0}, rotate={0.0}}, ymode={log}, log basis y={10}, ymajorgrids={true}, ymin={1e-9}, ymax={0.11435994055361993}, ytick={{1.0e-8,1.0e-6,0.0001,0.01}}, yticklabels={{$10^{-8}$,$10^{-6}$,$10^{-4}$,$10^{-2}$}}, ytick align={inside}, yticklabel style={font={{\fontsize{8 pt}{10.4 pt}\selectfont}}, color={rgb,1:red,0.0;green,0.0;blue,0.0}, draw opacity={1.0}, rotate={0.0}}, y grid style={color={rgb,1:red,0.0;green,0.0;blue,0.0}, draw opacity={0.1}, line width={0.5}, solid}, axis y line*={left}, y axis line style={color={rgb,1:red,0.0;green,0.0;blue,0.0}, draw opacity={1.0}, line width={1}, solid}, colorbar style={title={}}, point meta max={nan}, point meta min={nan}]
    \addplot[color=blue, name path={6bbdc228-316c-411a-b429-b7b55f11bff0}, draw opacity={1.0}, line width={1}, solid, mark={*}, mark size={3.0 pt}, mark options={color=black, draw opacity={1.0}, fill=blue, fill opacity={1.0}, line width={0.75}, rotate={0}, solid}, dashed]
        coordinates {
            (5,0.04282119599987606)
            (7,0.019696244977096777)
            (9,0.005847487980626263)
            (11,0.0014658958365942642)
            (13,0.00033165135178661294)
            (15,7.05002174038418e-5)
            (17,1.4258904952455394e-5)
            (19,2.992117374064235e-6)
            (21,6.098886235369962e-7)
        };
    \addlegendentry {SO-BT}
        ;
    \addplot[color=black, name path={6e2877b2-05c8-46fb-955e-6f02e81a55d9}, draw opacity={1.0}, line width={1}, solid, mark={diamond*}, mark size={3.0 pt}, mark options={color=black, draw opacity={1.0}, fill=black, fill opacity={1.0}, line width={0.75}, rotate={0}, solid}]
        coordinates {
            (5,0.037792735231474424)
            (7,0.07212023378374657)
            (9,0.007444890431391772)
            (11,0.0005119936905391195)
            (13,0.0013691365412874595)
            (15,0.0004128687908912796)
            (17,1.531939199575027e-5)
            (19,1.9340725602625436e-6)
            (21,5.420036531404618e-7)
        };
    \addlegendentry {$\hinf$-min approx.}
        ;
    \addplot[color=ourcolor, name path={33da23d1-c61f-40e9-bfb2-5a7a3ef72ca4}, draw opacity={1.0}, line width={1}, solid, mark={triangle*}, mark size={4.0 pt}, mark options={color=black, draw opacity={1.0}, fill=ourcolor, fill opacity={1.0}, line width={0.75}, rotate={0}, solid}]
        coordinates {
            (05,4.174850e-03)
            (07,3.298406e-04)
            (09,2.184062e-05)
            (11,2.779742e-07)
            (13,2.757395e-08)
            (15,1.230419e-08)
            (17,1.113622e-08)
            (19,1.115440e-08)
            (21,1.234850e-08)
        };
    \addlegendentry {SOBMOR}
        ;
    \addplot[color=ourcolor, name path={33da23d1-c61f-40e9-bfb2-5a7a3ef72ca4}, draw opacity={1.0}, line width={1}, solid, mark={square}, mark size={3.5 pt}, mark options={color=black, draw opacity={1.0}, fill=ourcolor, fill opacity={1.0}, line width={0.75}, rotate={0}, solid}, dashed]
        coordinates {
            (05, 4.174851e-03)
            (07, 3.298408e-04)
            (09, 2.184084e-05)
            (11, 2.771968e-07)
            (13, 3.025598e-08)
            (15, 2.727334e-08)
            (17, 1.118423e-08)
            (19, 6.121708e-09)
            (21, 2.041558e-08)
        };
        \addlegendentry {diagonal $M$}
\end{axis}
\end{tikzpicture}
  \caption{$\hinf$ errors for different MOR methods as ROM order $r$ is varied. The FOM is the second-order mechanical system with $\nx=301$. }%
  \label{fig:2nd_hinf_err}
\end{figure}
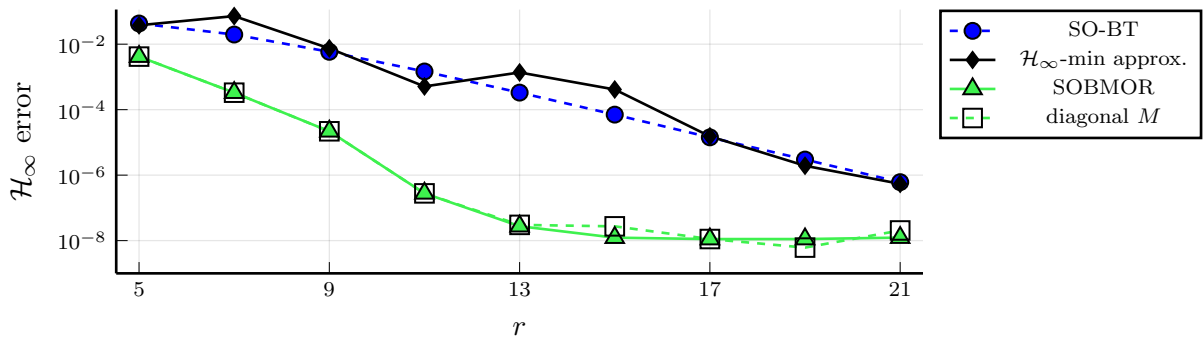

In Figure~\ref{fig:2nd_hinf_err}, we compare the $\hinf$ norm of the error transfer functions for ROMs with varying order. Note that our proposed method outperforms the other methods for structure-preserving model order reduction for all tested values of $r$ by a huge margin. To obtain a ROM with an $\hinf$ error of less than $10^{-6}$, using our method, a reduced model order of only $11$ is needed, while the other methods only reach this accuracy with a reduced model order of $21$. The results do not indicate a systematic difference in accuracy when using either a full or a diagonal parametrization of $M$, respectively. Furthermore, there is no significant difference in runtime, when using either parametrization.

It is interesting to note that beyond a reduced model order of $r=15$, the $\hinf$ error does not decrease significantly when our method is employed, which needs further investigation. We do not include a comparison with the direct $\hinf$ error minimization based on \texttt{ab13hd}, since the execution time of \texttt{ab13hd} on the error transfer function averages to around 20s, which renders the use of \texttt{ab13hd} inside an optimization loop computationally prohibitive.



Looking at the error transfer functions for different reduced orders which are depicted in Figure~\ref{fig:snd_errTF_methods}, it can be observed that the error using our method is nearly constant over a wide frequency range. The other methods lead to an error that is larger for almost all frequencies but also more oscillatory. Using our method, the error becomes less oscillatory for a larger reduced order.

\begin{figure}[tbh]
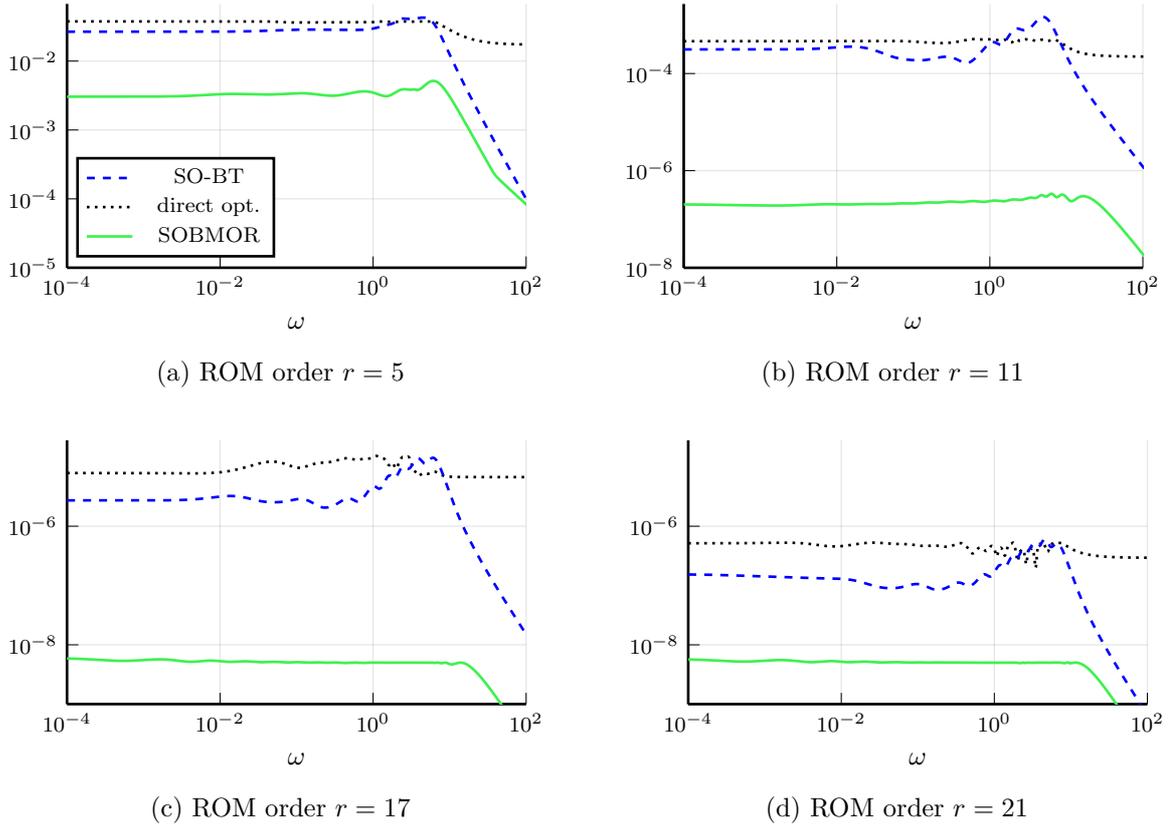

  \centering
  \begin{tabular}{cc}
    \input{./PlotSources/ErrorTransfunPlots/SND/MethodComparison/errors5.tex} &
    \input{./PlotSources/ErrorTransfunPlots/SND/MethodComparison/errors11.tex} \\
    (a) ROM order $r=5$ & (b) ROM order $r=11$ \vspace{0.5cm}\\ 
    \input{./PlotSources/ErrorTransfunPlots/SND/MethodComparison/errors17.tex} &
    \input{./PlotSources/ErrorTransfunPlots/SND/MethodComparison/errors21.tex}\\
    (c) ROM order $r=17$ & (d) ROM order $r=21$ \\
  \end{tabular}
  \caption{Maximum singular values of the error transfer function between the FOM (second-order mechanical system with dimension $\nx=301$) for different structure-preserving MOR methods as $r$ is varied.}%
  \label{fig:snd_errTF_methods}
\end{figure}

\section{Conclusions and Outlook}

In this paper, we have developed SOBMOR, a new approach for MOR based on direct optimization of the structured ROM's state-space realization. Instead of directly minimizing the $\hinf$ or $\htwo$ error, we have proposed to iteratively apply a leveled least squares objective function, which tends to lead to more accurate results at drastically reduced computational costs. We have demonstrated the benefits of our method and its superior performance compared to other structure-preserving MOR methods in several numerical experiments.

The method is applicable in a far broader setting that we plan to explore in subsequent works.
\begin{itemize}
  \item Until now, we have chosen logarithmically spaced sample points on the imaginary axis. However, we may decrease the computational footprint and necessary a priori knowledge about the critical intervals of the FOM by adaptive sampling. Therefore, in \cite{SchV2021Adaptive}, we propose an adaptive sampling algorithm, which removes the necessary a priori knowledge and even leads to smaller computation times.
  \item The method is entirely data-driven in that it does not require direct access to the system matrices of the full order model. This allows a modification of our proposed method to perform system identification for systems with a predefined structure. The effectiveness of this approach for passive systems is explored in \cite{Sch2021Ident}.
  \item We plan to explore the utilization of this method in parametric MOR in which the FOM depends on several parameters and the ROM should mimic the behavior of the FOM for all admissible parameter settings. 
\end{itemize}


\section*{Acknowledgement}
We thank Volker Mehrmann (TU Berlin) for his useful comments on an earlier version of this manuscript.

\bibliography{strucpres.bib}
\end{document}